\documentclass[12pt, draftclsnofoot, onecolumn]{IEEEtran}
\usepackage{amsfonts}
\usepackage{graphicx}
\usepackage{color}
\usepackage{amsmath,amsfonts,amssymb,amsthm,epsfig,epstopdf,url,array}
\usepackage{url,textcomp}
\usepackage{authblk}
\usepackage{algorithm}
\usepackage{algpseudocode}
\usepackage{varwidth}
\usepackage{caption,subcaption}
\usepackage{cite}
\newcommand{\bs}{\boldsymbol}
\newtheorem{theorem}{Theorem}

\makeatletter

\newcommand{\Rmnum}[1]{\expandafter\@slowromancap\romannumeral #1@}
\makeatother
\begin{document}
\title{Outage Performance and Optimal Design of MIMO-NOMA Enhanced Small Cell Networks With Imperfect Channel-State Information}
\author{
        Zheng~Shi,
        Hong~Wang,
        Yaru~Fu,
        Guanghua~Yang,
        Shaodan~Ma,
        and Xinrong~Ye
\thanks{Zheng Shi and Guanghua Yang are with the School of Intelligent Systems Science and Engineering, Jinan University, Zhuhai 519070, China, and also with the State Key Laboratory of Internet of Things for Smart City, University of Macau, Macao 999078, China
(e-mails: zhengshi@jnu.edu.cn, ghyang@jnu.edu.cn).}  
\thanks{Hong Wang is with the School of Communication and Information Engineering, Nanjing University of Posts and Telecommunications, Nanjing 210003, China, and also with the National Mobile Communications Research Laboratory, Southeast University, Nanjing 210096, China (e-mail: wanghong@njupt.edu.cn).}
\thanks{Yaru Fu is with the School of Science and Technology, The Open University of Hong Kong, Hong Kong SAR, China (e-mail: yfu@ouhk.edu.hk).}
\thanks{Shaodan Ma is with the State Key Laboratory of Internet of Things for Smart City and the Department of Electrical and Computer Engineering, University of Macau, Macao 999078, China (e-mail: shaodanma@um.edu.mo).}
\thanks{Xinrong~Ye is with the School of Physics and Electronic Information, Anhui Normal University, Wuhu 241002, China (e-mail: shuchong@mail.ahnu.edu.cn).}
}
\maketitle
\begin{abstract}
This paper focuses on boosting the performance of small cell networks (SCNs) by integrating multiple-input multiple-output (MIMO) and non-orthogonal multiple access (NOMA) in consideration of imperfect channel-state information (CSI). The estimation error and the spatial randomness of base stations (BSs) are characterized by using Kronecker model and Poisson point process (PPP), respectively. The outage probabilities of MIMO-NOMA enhanced SCNs are first derived in closed-form by taking into account two grouping policies, including random grouping and distance-based grouping. It is revealed that the average outage probabilities are irrelevant to the intensity of BSs in the interference-limited regime, while the outage performance deteriorates if the intensity is sufficiently low. Besides, as the channel uncertainty lessens, the asymptotic analyses manifest that the target rates must be restricted up to a bound to achieve an arbitrarily low outage probability in the absence of the inter-cell interference. Moreover, highly correlated estimation error ameliorates the outage performance under a low quality of CSI, otherwise it behaves oppositely. Afterwards, the goodput is maximized by choosing appropriate precoding matrix, receiver filters and transmission rates. In the end, the numerical results verify our analysis and corroborate the superiority of our proposed algorithm.
\end{abstract}
\IEEEpeerreviewmaketitle
\section{Introduction}\label{sec:int}
\IEEEPARstart{W}{ith} the surge of wireless data traffic and number of mobile devices, it has become increasingly imperative for 5G cellular networks to offer an improved spectrum efficiency as well as massive connectivity\cite{al2018small,sun2020towards}. As per Cisco visual networking index (VNI)\cite{cisco2018cisco},
5G speeds are estimated to be 13 times higher than the average mobile connection by 2023, and the number of mobile devices is forecasted to grow from 8.8 billion in 2018 to 13.1 billion by 2023. 
To confront these unprecedented challenges, small cell has been recognized as one of the promising solutions for 5G to deliver ever-increasing network capacity and fulfill the accommodation of a large number of users\cite{de2019key}. However, aggressive frequency reuse will incur severe interference for ultra-dense small cell deployment\cite{han2018downlink}. To remedy this issue, the fusion of small cell networks (SCNs) and other key enabling 5G technologies has received enormous interest\cite{han2018downlink,tan2017spectral,lei2019safeguarding}. Specifically, the combination of multiple-input multiple-output (MIMO) and non-orthogonal multiple access (NOMA) is anticipated to considerably improve spectral efficiency and support more concurrent connections for interference-infested SCNs\cite{wei2019performance,nasser2019interference,liu2020pilot,chen2020ambient}.
\subsection{Related Works and Motivations}
Since MIMO attains additional spatial diversity gain by deploying multiple antennas and NOMA exploits extra multi-user diversity gain by taking advantage of the channel discrepancy, it has been vastly reported that MIMO-NOMA achieves superior performance over its orthogonal counterpart, i.e., MIMO-orthogonal multiple access (OMA), in terms of both the reliability and the spectral efficiency \cite{ding2016application,ding2016general,wang2019power,fu2020zero}. Nevertheless, the application of MIMO-NOMA to SCNs still remains in its infancy. Notwithstanding, a wide range of investigations on single-input single-out (SISO)-NOMA enhanced SCNs have been conducted in the literature \cite{liu2017non,han2018downlink,zhang2019energy}. To mention only a few, by assuming single antenna at both small cell base stations (BSs) and users, the coverage probability and achievable rate of non-uniform SCNs were examined in \cite{han2018downlink}. By considering large-scale networks, it was studied in \cite{liu2017non} that the coverage probability of small cells depends heavily on the prescribed transmit rates and power sharing coefficients. In addition, a novel energy-efficient algorithm was proposed for NOMA aided heterogeneous SCNs with energy harvesting by assuming perfect channel-state information (CSI) in \cite{zhang2019energy}. Unfortunately, prior works seldom touched MIMO-NOMA enhanced SCNs except for \cite{wei2019performance,nasser2019interference}. In particular, the authors in \cite{wei2019performance} substantiated the superior performance of MIMO-NOMA over MIMO-OMA in multi-cell systems in terms of the ergodic sum-rate. Besides, a new interference mitigation and power allocation scheme was invented for downlink MIMO-NOMA assisted heterogeneous SCNs \cite{nasser2019interference}. However, the existing literature commonly assume perfect CSI at the receivers without the practical consideration of imperfect CSI in MIMO-NOMA enhanced SCNs.

Unfortunately, there are only few available methodologies that can be leveraged to undertake the performance investigations into MIMO-NOMA enhanced SCNs with imperfect CSI. To be specific, K. He \emph{et al.} developed a novel compressed sensing-aided MIMO-NOMA scheme by allowing for imperfect CSI in \cite{he2019novel}, where an upper bound for the performance of the proposed scheme was presented by using restricted isometry property. In \cite{silva2020noma}, the worst-case Gaussian approximation was employed to derive the total sum rate, with which the significant spectral and energy efficiency gains over the existing multi-way relay network were demonstrated. A semi-closed expression was provided for the outage performance evaluation of multiple-input single-output (MISO)-NOMA in view of imperfect CSI in \cite{chen2017low}, in which the assumption of correlated CSI was made. In \cite{cui2018outage}, Gaussian approximation was used to derive the outage probability of MIMO-NOMA systems under imperfect CSI. For the sake of analytical tractability, the uncorrelated estimation errors were assumed and only power allocation was involved into precoding without accounting for beamforming design in \cite{cui2018outage}. However, the developed methodologies in \cite{he2019novel,silva2020noma,chen2017low,cui2018outage} are based on either approximations or simulations, and they cannot be directly applied into the scenario of large-scale SCNs as well. Moreover, although a large-scale dense network was considered in \cite{chen2020performance}, the proposed analytical approach is specific to the multiple-input single-output (MISO)-NOMA system. In contrast, the performance analysis of the MIMO-NOMA system will introduce the difficulty of handling random matrix variate.


%
%
%
%
%

This motivates us to trigger the performance investigations of MIMO-NOMA enhanced large-scale SCNs with imperfect CSI at both the transmit and receive sides. More specifically, Poisson point process (PPP) and Kronecker model are utilized to capture the randomness of small cell BSs' positions and the impact of estimation error correlation, respectively. By employing stochastic geometry and integral transform, the exact outage probabilities are derived in closed-form and insightful results are extracted by conducting asymptotic outage analyses. With the analytical results, the optimal design of MIMO-NOMA enhanced SCNs is empowered. To the best of our knowledge, this is the first work that copes with MIMO-NOMA enhanced SCNs under imperfect CSI. Additionally, in contrast to \cite{ding2016general,ding2016application,cui2018outage}, a more flexible MIMO-NOMA transmission framework with variable number of clusters/groups is established in this paper. Furthermore, it is worth noting that the Gaussian approximation used in \cite{cui2018outage} cannot be applied in this paper because of the nonexistence of mean and variance of the aggregated inter-cell interference.
\subsection{Contributions}
To recapitulate, five main contributions of this paper are enumerated below.
\begin{enumerate}
  \item Exact compact expressions are derived for the outage probabilities of MIMO-NOMA enhanced SCNs, where two NOMA grouping policies are considered, namely random NOMA grouping and distanced-based NOMA grouping.
  \item The exact outage analyses show that the average outage probability is irrespective of the intensity of BSs in the interference-limited regime. Whereas, a counterintuitive behavior that the intensity of BSs below a threshold is harmful to the average outage performance is observed.
  \item As the channel uncertainty vanishes, the target transmission rates should be upper bounded to guarantee an arbitrarily low outage probability in the absence of the inter-cell interference.
  \item The correlation coefficient has a positive effect on the outage probability under the high quality of CSI. Conversely, the correlation coefficient impairs the outage performance if the quality of CSI is low.
  \item The goodput maximization of the MIMO-NOMA system is considered by jointly devising the precoding matrix, receiver filters and transmission rates. Numerical results exhibit that the proposed scheme outperforms other benchmarking schemes especially under significant difference between channel gains, including MIMO-OMA with precoding, MIMO-OMA without precoding and MIMO-NOMA without precoding.
\end{enumerate}

\subsection{Paper Organization and Notation}
The remaining of the paper is organized as follows. Section \ref{sec:sys_mod} delineates the model of MIMO-NOMA enhanced SCNs. The exact and asymptotic outage analyses are carried out in Sections \ref{sec:out} and \ref{sec:asy}, respectively. Section \ref{sec:opt} devotes itself to the optimal system design. Numerical results are then presented for validations in Section \ref{label:num}. Section \ref{sec:con} finally concludes this paper.

\emph{Notation}: The following notations shall be adopted throughout this paper. The uppercase and lowercase boldface letters denote matrices and vectors, respectively. ${\bf X}^{\mathrm{T}}$, ${\bf X}^{\mathrm{H}}$, ${\bf X}^{-1}$ and ${\bf X}^{1/2}$ stands for the transpose, conjugate transpose, matrix inverse and Hermitian square root of matrix ${\bf X}$, respectively. ${\rm{Tr}}(\cdot)$, ${\mathrm{det}}(\cdot)$ and ${\mathrm{diag}}(\cdot)$ refer to trace, determinant and diagonalization operations, respectively. $\left\| \cdot \right\|$ represents the Euclidean/Frobenius norm of a vector/matrix. ${{\mathbb{E}}_A}\{\cdot\}$ denotes the operator of the expectation taking over random variable $A$, and the subscription $A$ is sometimes omitted without leading to ambiguity. $\mathbb C^{m\times n}$ is the set of $m\times n$-dimensional complex matrices. ${\cal CN}({\bs \mu, \bs \Sigma})$ denotes the distribution of a circularly-symmetric complex Gaussian (CSCG) random variables with mean vector $\bs \mu$ and covariance matrix $\bs \Sigma$. $\mathbf{0}_n$, $\mathbf{1}_n$ and $\mathbf{I}_n$ represent an $n \times 1$ all-zeros vector, an $n \times 1$ all-ones vector and an $n \times n$ identity matrix, respectively, and the dimensionality $n$ is sometimes dropped out for convenience. The symbol ${\rm i}=\sqrt{-1}$ refers to the imaginary unit. ${\rm Re}\{\cdot\}$ represents the real part of a complex number. $\bigcup{(\cdot)}$ stands for the union operation. $H^{m,n}_{p,q}(\left. x \right|\cdot)$ represents Fox's H-function. Any other notations are deferred to define in the place where they appear.

\section{System Model}\label{sec:sys_mod}
This paper considers a MIMO-NOMA enhanced SCNs with imperfect CSI. We assume that the positions of the small cell BSs obey a homogeneous PPP $\Phi_b$ of intensity $\lambda_b$ and the users are also distributed according to a homogenous PPP $\Phi_u$ of intensity $\lambda_u$. Each user is associated with its closest BS\cite{salehi2019accuracy,park2012outage}. The BSs and the users are equipped with $M$ and $N$ antennas, respectively. By considering heavily loaded networks, i.e., $\lambda_u \gg \lambda_b$, it is reasonable to assume that there are at least $2K$ users in a typical Voronoi cell and $K \le \min\{M,N\}$\cite{haenggi2017user,yu2013downlink}. To accommodate $K$ pairs of NOMA users over MIMO channels within the coverage of each cell, $2K$ users can be randomly chosen in the Voronoi cell. Accordingly, the NOMA users are uniformly distributed within the Voronoi cell \cite{salehi2019accuracy}. In \cite{yu2012dynamic,haenggi2017user,wang2017meta}, the distance between a user and its serving BS approximately follows a Rayleigh distribution with probability density function (PDF)
\begin{equation}\label{eqn:pdf_d_approx}
  f_d(x) \approx 2c\lambda_b\pi x e^{-c\lambda_b\pi x^2},\, x\ge 0,
\end{equation}
where $c=5/4$. Hereinafter, the system model is introduced by splitting into three subsections, including transmission model, imperfect CSI and received SINR.
\subsection{Transmission Model}
In this paper, the NOMA scheme is implemented in the power-domain. By following the similar MIMO-NOMA framework to \cite{ding2016general}, the $K \times 1$ information-bearing vector ${{\bf{s}}}$ sent by a typical BS $z$ is expressed as
\begin{equation}\label{eqn:tran_signal}
{{\bf{s}}} = \left(
{{\beta _{1}}{s_{1}} + {\beta _{\tilde 1}}{s_{\tilde 1}}},\cdots,{{\beta _{k}}{s_{k}} + {\beta _{\tilde k}}{s_{\tilde k}}},\cdots,{{\beta _{K}}{s_{K}} + {\beta _{\tilde K}}{s_{\tilde K}}}
 \right)^{\rm T},
\end{equation}
where ${s_{k}}$ and ${s_{\tilde k}}$ are the signals intended for the near user and the far user of the $k$-th pair, respectively, and they are normalized to unit power, i.e., ${\mathbb E}\{|s_k|^2\} = {\mathbb E}\{|s_{k'}|^2\} \triangleq 1$, ${\beta _{k}}$ and ${\beta _{\tilde k}}$ correspond to the power allocation coefficients with ${\beta _{k}}^2 + {\beta _{\tilde k}}^2 = 1$. 
The received signal at user $k$ (or $\tilde k$) can be expressed as
\begin{align}\label{eqn:users_near_signal_rec}
{\bf{y}}_k &= {\sqrt{P\ell \left( {\left\| {z - {o_k}} \right\|} \right)}}{{\bf{H}}}_{zk}{{\bf{V}}}{{\bf{s}}} 
+ {\sum\nolimits _{x \in {\Phi _b}\backslash \left\{ z \right\}}}\sqrt {\rho_I\ell \left( {\left\| {x - {o_k}} \right\|} \right)} {{\bf{1}}_N}{w_x}  + {\bf{n}}_k,
\end{align}
where ${o_k}$ stands for the position of user $k$, $P$ denotes the transmit power of each data stream, ${{\bf{H}}}_{zk} \in {\mathbb C}^{N \times M}$ represents the channel response matrix from the BS $z$ to user $k$, $x \in {\Phi _b}\backslash \left\{ z \right\}$ stands for the interfering BSs (the co-channel BSs excluding $z$) and ${{{\bf{V}}}}=\left( {{{\bf{v}}_{1}}, \cdots ,{{\bf{v}}_{K}}} \right)$ is the normalized $M \times K$ precoding matrix, i.e., $\left\| {\bf v}_k \right\|=1$ for $k =1,\cdots,K$, the co-channel interference is modeled according to the classical shot noise model due to the unavailable CSI between interfering BS $x \in {\Phi _b}$ and user $k$ \cite{ding2016general}, $\ell \left( d \right) = {d^{ -\alpha }}$ captures path loss and $\alpha$ stands for the path loss exponent, ${{ w}_{{x}}}$ denotes the normalized signal sent by the interferer, $\rho_I$ corresponds to the interference power, ${\bf n}_k$ denotes the complex-valued additive Gaussian white noise with variance ${\sigma^2}{\bf I}_M$.
\subsection{Imperfect CSI}
We assume that only partial CSI is available at both the transmitter and the receiver. More precisely, the imperfect channel model for ${\bf H}_{zk}$ is formulated as \cite{park2012outage}
\begin{equation}\label{eqn:imperfect_channel_model}
 {\bf H}_{zk} = \hat {\bf H}_{zk} + {\bf E}_{zk}.
\end{equation}
where $\hat {\bf H}_{zk}$ is the channel state known to BS $z$ and user $k$, and ${\bf E}_{zk}$ is the channel estimation error. Furthermore, a frequently used error model is adopted for the estimation error ${\bf E}_{zk}$, i.e., Kronecker correlation model. In particular, the channel uncertainty matrix ${\bf E}_{zk}$ is given by
\begin{equation}\label{eqn:imper_csi_konecker}
{\bf{E}}_{zk} = {{\bf{R }}_{rk}}^{1/2}{ {\bf{E}}_w}{{\bf{R }}_{tk}}^{1/2}.
\end{equation}
where ${\rm vec}({\bf E}_w) \sim {\cal CN}({\bf 0},{\sigma_h}^2{\bf I})$. It is worth noting that error correlations take place because of the transmit and receive antenna structure. Accordingly, ${\bf{E}}_{zk}$ is a CSCG random matrix, i.e., ${\rm vec}({\bf E}_{zk}) \sim {\cal CN}({\bf 0},{\sigma_h}^2({{\bf R}_{tk}^{\rm T}\otimes{\bf R}_{rk}}))$. It is assumed that the partial channel informations $\hat{\bf H}_{zk}$, ${\bf R}_{tk}$ and ${\bf R}_{rk}$ are known to BS $z$ and user $k$. To characterize the quality of CSI, we introduce the channel K factor defined as \cite{park2012outage} 
\begin{equation}\label{eqn:kch}
  \mathcal K_{zk} = \frac{\| {\hat {\bf H}_{zk}} \|_F^2}{\mathbb E\{\left\| {{\bf E}_{zk}} \right\|_F^2\}} 
  = \frac{\| {\hat {\bf H}_{zk}} \|_F^2}{{\sigma_h}^2{\rm Tr}{({{\bf R}_{tk}})}{\rm Tr}({\bf R}_{rk})},
\end{equation}
where $\mathcal K_{zk}$ represents the power ratio of the known channel part to the unknown part, which quantifies the quality of CSI. 
In addition, the similar definition also applies to the imperfect CSI associated with user $\tilde k$.
%
\subsection{Received SINR}
Similarly to \cite{ding2016general}, the superimposed signals at the receive sides are successively decoded in a descending order of link distances. Specifically, with regard to the near user $k$, by applying receiver filter ${\bf u}_k$ to \eqref{eqn:users_near_signal_rec}, the desired data stream $k$ can be written as
\begin{align}\label{eqn:y_m_wr}
{{\hat s}_k}=&\sqrt {P\ell \left( {{d_k}} \right)} {{\bf{u}}_k}^{\rm{H}}\left( {{{{\bf{\hat H}}}_{zk}} + {{\bf{E}}_{zk}}} \right){{\bf{v}}_k}\left( {{\beta _k}{s_k} + {\beta _{\tilde k}}{s_{\tilde k}}} \right) 
 \notag\\
&+\underbrace {\sqrt {P\ell \left( {{d_k}} \right)} \sum\limits_{i = 1,i \ne k}^K {{{\bf{u}}_k}^{\rm{H}}\left( {{{{\bf{\hat H}}}_{zk}} + {{\bf{E}}_{zk}}} \right){{\bf{v}}_i}\left( {{\beta _i}{s_i} + {\beta _{\tilde i}}{s_{\tilde i}}} \right)} }_{{\rm{intra - cell}}\,{\rm{interference}}}\notag\\
& + \underbrace {{\sum _{x \in {\Phi _b}\backslash \left\{ z \right\}}}\sqrt {{\rho _I}\ell \left( {\left\| {x - {o_k}} \right\|} \right)} {{\bf{u}}_k}^{\rm{H}}{{\bf{1}}_N}{w_x}}_{{\rm{inter-cell\,interference}}} + {{\bf{u}}_k}^{\rm{H}}{{\bf{n}}_k}.
\end{align}
where $\left\| {z - {o_k}} \right\| = d_k$.
By following the NOMA principle, the technique of successive interference cancellation (SIC) is adopted to cancel out the message intended to user $\tilde k$ with SINR given by 
\begin{align}\label{eqn:gamma_sinr_mdash}
{{\rm SINR} _{k \to \tilde k}} = \frac{{P\ell \left( {{d_k}} \right){{\left| {{{\bf{u}}_k}^{\rm{H}}{{{\bf{\hat H}}}_{zk}}{{\bf{v}}_k}} \right|}^2}{\beta _{\tilde k}}^2}}{{P\ell \left( {{d_k}} \right)\left( \begin{array}{c}
{\left| {{{\bf{u}}_k}^{\rm{H}}{{\bf{E}}_{zk}}{{\bf{v}}_k}} \right|^2}{\beta _{\tilde k}}^2 + {\left| {{{\bf{u}}_k}^{\rm{H}}\left( {{{{\bf{\hat H}}}_{zk}} + {{\bf{E}}_{zk}}} \right){{\bf{v}}_k}} \right|^2}{\beta _k}^2\\
 + \sum\nolimits_{i=1, i \ne k}^{K} {{{\left| {{{\bf{u}}_k}^{\rm{H}}\left( {{{{\bf{\hat H}}}_{zk}} + {{\bf{E}}_{zk}}} \right){{\bf{v}}_i}} \right|}^2}}
\end{array} \right) + {I_k} + {\sigma_k ^2}}},
\end{align}
where $\sigma_k ^2=\sigma^2\left\| {{{\bf{u}}_k}} \right\|^2$ and ${I_k}$ is the power of the inter-cell interference observed by user $k$, such that
\begin{equation}\label{eqn:interfere_near}
{I_k} = {\rho _I}{\left| {{{\bf{u}}_k}^{\rm{H}}{{\bf{1}}_N}} \right|^2}{\sum\nolimits _{x \in {\Phi _b}\backslash \left\{ z \right\}}}\ell \left( {\left\| {x - {o_k}} \right\|} \right).
\end{equation}
By subtracting the message $s_{\tilde k}$, the near user decodes its own message $s_{k}$ with SINR given by 
\begin{align}\label{eqn:sinr_mm}
{{\rm SINR} _k} = \frac{{P\ell \left( {{d_k}} \right){{\left| {{{\bf{u}}_k}^{\rm{H}}{{{\bf{\hat H}}}_{zk}}{{\bf{v}}_k}} \right|}^2}{\beta _k}^2}}{{P\ell \left( {{d_k}} \right)\left( {{{\left| {{{\bf{u}}_k}^{\rm{H}}{{\bf{E}}_{zk}}{{\bf{v}}_k}} \right|}^2} + \sum\nolimits_{i=1, i \ne k}^{K} {{{\left| {{{\bf{u}}_k}^{\rm{H}}\left( {{{{\bf{\hat H}}}_{zk}} + {{\bf{E}}_{zk}}} \right){{\bf{v}}_i}} \right|}^2}} } \right) + {I_k} + {\sigma_{ k} ^2}}}.
\end{align}
In an analogous way, the far user $\tilde k$ decodes its own message with SINR given by
\begin{equation}\label{eqn:sinr_mdashmdash}
{{\rm SINR} _{\tilde k}} = \frac{{P\ell \left( {{d_{\tilde k}}} \right){{\left| {{{\bf{u}}_{\tilde k}}^{\rm{H}}{{{\bf{\hat H}}}_{z\tilde k}}{{\bf{v}}_k}} \right|}^2}{\beta _{\tilde k}}^2}}{{P\ell \left( {{d_{\tilde k}}} \right)\left( \begin{array}{c}
{\left| {{{\bf{u}}_{\tilde k}}^{\rm{H}}{{\bf{E}}_{z\tilde k}}{{\bf{v}}_k}} \right|^2}{\beta _{\tilde k}}^2 + {\left| {{{\bf{u}}_{\tilde k}}^{\rm{H}}\left( {{{{\bf{\hat H}}}_{z\tilde k}} + {{\bf{E}}_{z\tilde k}}} \right){{\bf{v}}_k}} \right|^2}{\beta _k}^2\\
 + \sum\nolimits_{i=1,i \ne m}^K {{{\left| {{{\bf{u}}_{\tilde k}}^{\rm{H}}\left( {{{{\bf{\hat H}}}_{z\tilde k}} + {{\bf{E}}_{z\tilde k}}} \right){{\bf{v}}_i}} \right|}^2}}
\end{array} \right) + {I_{\tilde k}} + {\sigma_{\tilde k} ^2}}},
\end{equation}
where ${\sigma_{\tilde k} ^2} = \sigma^2\left\| {{{\bf{u}}_{\tilde k}}} \right\|^2$ and ${\bf u}_{\tilde k}$ refers to the detection filter applied at user $\tilde k$, $\left\| {z - {o_{\tilde k}}} \right\| = d_{\tilde k}$ and the interference term ${I_{\tilde k}}$ reads as
\begin{equation}\label{eqn:I_mdash}
{I_{\tilde k}} = {\rho _I}{\left| {{{\bf{u}}_{\tilde k}}^{\rm{H}}{{\bf{1}}_N}} \right|^2}{\sum\nolimits _{x \in {\Phi _b}\backslash \left\{ z \right\}}}\ell \left( {\left\| {x - {o_{\tilde k}}} \right\|} \right).
\end{equation}
It is worth noting that the aggregated interferences ${I_k}$ and ${I_{\tilde k}}$ are correlated due to the same interfering BSs, i.e., $x\in \Phi_b\backslash \{z\}$.
\section{Outage Probability}\label{sec:out}
Since the outage probability is the key performance metric, this section is dedicated to conducting the outage analysis for MIMO-NOMA enhanced SCNs with imperfect CSI. The outage probabilities of users $\tilde k$ and $k$ are separately derived as follows. 
%
%
%
\subsection{Outage Probability of the Far User $\tilde k$}
\subsubsection{Conditional Outage Probability}\label{sec:outkt}
The outage event at the far user $\tilde k$ happens if the channel capacity is less than the preset transmission rate $R_{\tilde k}$. Given the collection of channel information $\mathcal H = \{\hat {\bf H}_{zk},{\bf R}_{rk}, {\bf R}_{tk}, \hat {\bf H}_{z\tilde k},{\bf R}_{r\tilde k}, {\bf R}_{t\tilde k}, d_k, d_{\tilde k}, \alpha, \sigma^2, {\sigma_h}^2, \rho_I, \lambda_b\}$, the conditional outage probability of user $\tilde k$ can be expressed as
\begin{align}\label{eqn:out_mdash}
{p_{\tilde k|\mathcal H}} &= {\Pr \left( {{{\log }_2}\left( {1 + {{\rm SINR} _{\tilde k}}} \right) < {R_{\tilde k}}} \right)} 
=1- \underbrace{\Pr \left( {{{\log }_2}\left( {1 + {{\rm SINR} _{\tilde k}}} \right) \ge {R_{\tilde k}}} \right)}_{{q_{\tilde k|{\mathcal H}}}},
\end{align}
where ${{q_{\tilde k|{\mathcal H}}}}$ stands for the conditional successful probability given the statistical channel information $\mathcal H $. 
We thus proceed to derive ${{q_{\tilde k|{\mathcal H }}}}$ next. By using \eqref{eqn:sinr_mdashmdash}, it follows that
\begin{align}\label{eqn:out_kd_tilde}
&{q_{\tilde k|{\mathcal H}}} =
\Pr \left\{ \begin{array}{l}
{\beta _{\tilde k}}^2{\left| {{\chi _{\tilde k}}} \right|^2} + {\beta _k}^2{\left| {{\mu _{\tilde k}} + {\chi _{\tilde k}}} \right|^2} + \sum\limits_{i = 1,i \ne m}^K {{{\left| {{\mu _{\tilde i}} + {\chi _{\tilde i}}} \right|}^2}}
 \le \frac{{{{\left| {{\mu _{\tilde k}}} \right|}^2}{\beta _{\tilde k}}^2}}{{{2^{{R_{\tilde k}}}} - 1}} - \frac{{{{\sigma_{\tilde k} ^2}}}}{{P\ell \left( {{d_{\tilde k}}} \right)}} - \frac{{{I_{\tilde k}}}}{{P\ell \left( {{d_{\tilde k}}} \right)}}
\end{array} \right\},
\end{align}
where ${\mu _{\tilde i}} = {{\bf{u}}_{\tilde k}}^{\rm{H}}{{{\bf{\hat H}}}_{z\tilde k}}{{\bf{v}}_i}$ and ${\chi_{\tilde i}} = {{\bf{u}}_{\tilde k}}^{\rm{H}}{{\bf{E}}_{z\tilde k}}{{\bf{v}}_i}$. Then, by using the identity ${\beta _{\tilde k}}^2{\left| {{\chi_{\tilde k}}} \right|^2} + {\beta _k}^2{\left| {{\mu _{\tilde k}} + {\chi_{\tilde k}}} \right|^2} = {\left| {{\beta _k}^2{\mu _{\tilde k}} + {\chi_{\tilde k}}} \right|^2} + {\beta _k}^2{\beta _{\tilde k}}^2{\left| {{\mu _{\tilde k}}} \right|^2}$,
\eqref{eqn:out_kd_tilde} can be simplified as
\begin{align}\label{eqn:p_out_cnd_far_simp}
{q_{\tilde k|{\mathcal H}}} &= \Pr \left\{ {{{\left( {{\bs{\tilde {\chi}}} + {\bs{\tilde \nu }}} \right)}^{\rm{H}}}\left( {{\bs{\tilde {\chi}}} + {\bs{\tilde \nu }}} \right) \le {\tau _{\tilde k}} - \frac{{{I_{\tilde k}}}}{{P\ell \left( {{d_{\tilde k}}} \right)}}} \right\},
\end{align}
where ${\bs{\tilde {\chi}}} = {\left( {{{\chi}_{\tilde 1}}, \cdots ,{{\chi}_{\tilde K}}} \right)^{\rm{T}}} = {\left( {{{\bf{u}}_{\tilde k}}^{\rm{H}}{{\bf{E}}_{z\tilde k}}{\bf{V}}} \right)^{\rm{T}}}$, ${\tau _{\tilde k}} = \left( {{1}/{({{2^{{R_{\tilde k}}}} - 1})} - {\beta _k}^2} \right){\beta _{\tilde k}}^2{\left| {{\mu _{\tilde k}}} \right|^2} - {{{\sigma_{\tilde k} ^2}}}/{({P\ell \left( {{d_{\tilde k}}} \right)})}$ and ${\bs{\tilde \nu }} = {\left( {{\mu _{\tilde 1}}, \cdots ,{\mu _{\tilde{k - 1}}},{\beta _k}^2{\mu _{\tilde k}},{\mu _{\tilde{k + 1}}}, \cdots ,{\mu _{\tilde K}}} \right)^{\rm{T}}}$.
In order to guarantee ${q_{\tilde k|{\mathcal H}}}<1$ from \eqref{eqn:p_out_cnd_far_simp}, the condition ${\beta _k}^2\left( {{2^{{R_{\tilde k}}}} - 1} \right) < 1$ should be satisfied. Apparently, ${\bs{\tilde {\chi}}}$ is a CSCG random vector with mean zero and covariance matrix ${\bf{\tilde \Sigma }} = {\sigma _h^2}{{\bf{u}}_{\tilde k}}^{\rm{H}}{{\bf{R}}_{r \tilde k}}{{\bf{u}}_{\tilde k}}{\left( {{{\bf{V}}^{\rm{H}}}{{\bf{R}}_{t\tilde k}}{\bf{V}}} \right)^{\rm{T}}}$, i.e., ${\bs{\tilde {\chi}}}\sim {\cal CN}(\bf 0 ,\tilde{\bf\Sigma} )$\cite[eq.13]{park2012outage}. In analogous to characteristic function used in \cite{al2016distribution}, applying inverse Laplace transform to \eqref{eqn:p_out_cnd_far_simp} yields 
\begin{align}\label{eqn:inve_Lap_out_tilde}
q_{\tilde k|{\mathcal H}} &= {{\mathbb{E}}_{{I_{\tilde k}}}}\left\{ {\int\limits_{{{\mathbb C}^K}} {u\left( {{\tau _{\tilde k}} - \frac{{{I_{\tilde k}}}}{{P\ell \left( {{d_{\tilde k}}} \right)}}}-{{\left( {{\bf{x}} + {\bs{\tilde \nu }}} \right)}^{\rm{H}}}\left( {{\bf{x}} + {\bs{\tilde \nu }}} \right) \right){f_{{\bs{\tilde \chi}}}}\left( {\bf{x}} \right)} d{\bf{x}}} \right\}\notag\\
 &= {{\mathbb{E}}_{{I_{\tilde k}}}}\left\{ {\int\limits_{{{\mathbb C}^K}} {\frac{1}{{2\pi {\rm{i}}}}\int\limits_{\tilde c - {\rm{i}}\infty }^{\tilde c + {\rm{i}}\infty } {\frac{1}{s}{e^{s\left( {{\tau _{\tilde k}} - \frac{{{I_{\tilde k}}}}{{P\ell \left( {{d_{\tilde k}}} \right)}}}-{{\left( {{\bf{x}} + {\bs{\tilde \nu }}} \right)}^{\rm{H}}}\left( {{\bf{x}} + {\bs{\tilde \nu }}} \right)\right)}}ds} {f_{{\bs{\tilde \chi}}}}\left( {\bf{x}} \right)} d{\bf{x}}} \right\},\,\tilde c>0,
\end{align}
where the last step holds by using $u(x) = \frac{1}{{2\pi {\rm{i}}}}\int_{\tilde c - {\rm{i}}\infty }^{\tilde c + {\rm{i}}\infty } {{1}/{s}{e^{sx}}ds}$, $u(x)$ stands for the step unit function and ${f_{{\bs{\tilde \chi}}}}\left( {\bf{x}} \right)$ represents the joint PDF of ${\bs{\tilde \chi}}$. Since the joint PDF of ${\bs{\tilde \chi}}$ is ${f_{\bs{\tilde \chi}}}\left( {\bf{x}} \right) ={\exp({ - {{\bf{x}}^{\rm{H}}}{{\bf{\tilde\Sigma }}^{ - 1}}{\bf{x}}})}/{({{\pi ^K}\det  ({\bf{\tilde\Sigma }}) })}$,
\eqref{eqn:inve_Lap_out_tilde} can be further rewritten as
\begin{align}\label{eqn:p_inte_out1}
q_{\tilde k|{\mathcal H}} =& \frac{1}{{{\pi ^K}\det \left( {{\bf{\tilde \Sigma }}} \right)}}\frac{1}{{2\pi {\rm{i}}}}\int\limits_{\tilde c - {\rm{i}}\infty }^{\tilde c + {\rm{i}}\infty } {\frac{{{e^{ s{\tau _{\tilde k}}}}}}{s}{{\mathbb{E}}_{{I_{\tilde k}}}}\left\{ {{e^{-s\frac{{{I_{\tilde k}}}}{{P\ell \left( {{d_{\tilde k}}} \right)}}}}} \right\}ds}
\int\limits_{{{\mathbb C}^K}} {{e^{-s{{\left( {{\bf{x}} + {\bs{\tilde \nu }}} \right)}^{\rm{H}}}\left( {{\bf{x}} + {\bs{\tilde \nu }}} \right) - {{\bf{x}}^{\rm{H}}}{{{\bf{\tilde \Sigma }}}^{ - 1}}{\bf{x}}}}} d{\bf{x}}.
\end{align}
By using the sum of two quadratic forms \cite[eq.8.1.7]{petersen2008matrix} and the important integration \cite[eq.22]{al2016distribution}, the inner integral in \eqref{eqn:p_inte_out1} can be computed as
\begin{align}\label{eqn:identi_gassian_inte}
\int\limits_{{\mathbb C^K}} {{e^{-s{{\left( {{\bf{x}} + {\bs{\tilde \nu }}} \right)}^{\rm{H}}}\left( {{\bf{x}} + {\bs{\tilde \nu }}} \right) - {{\bf{x}}^{\rm{H}}}{{{\bf{\tilde \Sigma }}}^{ - 1}}{\bf{x}}}}} d{\bf{x}} 
&= \int\limits_{{\mathbb C^K}} {{e^{{{-\left( {{\bf{x}} + {\bf{\tilde a}}} \right)}^{\rm{H}}}\left( {s{\bf{I}} + {{{\bf{\tilde \Sigma }}}^{ - 1}}} \right)\left( {{\bf{x}} + {\bf{\tilde a}}} \right) - {{{\bs{\tilde \nu }}}^{\rm{H}}}{{\left( {\frac{1}{s}{\bf{I}} + {\bf{\tilde \Sigma }}} \right)}^{ - 1}}{\bs{\tilde \nu }}}}d{\bf{x}}} \notag\\
&= \frac{{{\pi ^K}}}{{\det \left( { s{\bf{I}} + {{{\bf{\tilde \Sigma }}}^{ - 1}}} \right)}}{e^{-{{{\bs{\tilde \nu }}}^{\rm{H}}}{{\left( {\frac{1}{s}{\bf{I}} + {\bf{\tilde \Sigma }}} \right)}^{ - 1}}{\bs{\tilde \nu }}}}.
\end{align}
Furthermore, by using the Laplace functional of ${I_{\tilde k}}$, the expectation term in the above can be obtained as  \cite[eq.5.11]{haenggi2012stochastic}
\begin{align}\label{eqn:laplacefunc_faruser}
{{\mathbb{E}}_{{I_{\tilde k}}}}\left\{ {{e^{-s\frac{{{I_{\tilde k}}}}{{P\ell \left( {{d_{\tilde k}}} \right)}}}}} \right\}&= {e^{ - \pi {\lambda _b}{\omega _{\tilde k}}{d_{\tilde k}}^2{s^{\frac{2}{\alpha }}}}},
\end{align}
where ${\omega _{\tilde k}} = \Gamma \left( {1 - \frac{2}{\alpha }} \right){\left( {\frac{{{\rho _I}}}{P}{{\left| {{{\bf{u}}_{\tilde k}}^{\rm{H}}{{\bf{1}}_K}} \right|}^2}} \right)^{\frac{2}{\alpha }}}$. Thus, by substituting \eqref{eqn:identi_gassian_inte} and \eqref{eqn:laplacefunc_faruser} into \eqref{eqn:p_inte_out1}, we obtain
\begin{align}\label{eqn:p_inte_out2}
q_{\tilde k|{\mathcal H}} 
&= \frac{1}{{\det \left( {{\bf{\tilde \Sigma }}} \right)}}
 \frac{1}{{2\pi {\rm{i}}}}\int\limits_{\tilde c - {\rm{i}}\infty }^{\tilde c + {\rm{i}}\infty } {\frac{{{e^{-{{{\bs{\tilde \nu }}}^{\rm{H}}}{{\left( {\frac{1}{s}{\bf{I}} + {\bf{\tilde \Sigma }}} \right)}^{ - 1}}{\bs{\tilde \nu }} + {\tau _{\tilde k}}s{ - \pi {\lambda _b}{\omega _{\tilde k}}{d_{\tilde k}}^2{s^{\frac{2}{\alpha }}}}}}}}{{s\det \left( { s{\bf{I}} + {{{\bf{\tilde \Sigma }}}^{ - 1}}} \right)}}ds}.
\end{align}
With the eigenvalue decomposition ${\bf{\tilde \Sigma }}=\bs{\tilde \Psi}  \bs{\tilde \Delta } {\bs{\tilde \Psi}}^{\rm{H}}$, \eqref{eqn:p_inte_out2} reduces to
\begin{equation}\label{eqn:p_inte_out2_simdef}
q_{\tilde k|{\mathcal H}} = \frac{1}{{2\pi {\rm{i}}}}\int\limits_{\tilde c - {\rm{i}}\infty }^{\tilde c + {\rm{i}}\infty } {{e^{{\tau _{\tilde k}}s}}\underbrace {\frac{{\prod\nolimits_{i = 1}^K {{e^{ - \frac{{s{{\left| {{\zeta _{\tilde i}}} \right|}^2}}}{{1 + s{\delta _{\tilde i}}}}}}} {e^{ - \pi {\lambda _b}{\omega _{\tilde k}}{d_{\tilde k}}^2{s^{\frac{2}{\alpha }}}}}}}{{s\prod\nolimits_{i = 1}^K {\left( {1 + s{\delta _{\tilde i}}} \right)} }}}_{{F_{\tilde k}}\left( s \right)}ds} ,
\end{equation}
where $\bs{\tilde \Delta}  = {\rm{diag}}\left( {{\delta _{\tilde 1}}, \cdots ,{\delta _{\tilde K}}} \right)$ and $\bs{\tilde \Psi}^{\rm{H}} \bs{\tilde \nu}  = {\left( {{\zeta  _{\tilde 1}}, \cdots ,{\zeta _{\tilde K}}} \right)^{\rm{H}}}$. In order to calculate $q_{\tilde k|{\mathcal H}}$, a numerical inversion of Laplace transform developed in \cite{abate1995numerical} can be applied herein.
More specifically, by using the method of Abate and Whitt, \eqref{eqn:p_inte_out2_simdef} can be approximated with an arbitrarily small discretization error as\cite{abate1995numerical}
\begin{align}\label{eqn:lap_approx}
q_{\tilde k|{\mathcal H}} \approx &\frac{{{2^{ - M}}{e^{A/2}}}}{{{\tau _{\tilde k}}}}\sum\limits_{m = 0}^M {{{M}\choose{m}}}
 \left( \begin{array}{l}
\frac{1}{2}{\rm{Re}}\left\{ {{F_{\tilde k}}\left( {\frac{A}{{2{\tau _{\tilde k}}}}} \right)} \right\} +
\sum\limits_{n = 1}^{Q + m} {{{\left( { - 1} \right)}^n}{\rm{Re}}\left\{ {{F_{\tilde k}}\left( {\frac{{A + 2n\pi {\rm{i}}}}{{2{\tau _{\tilde k}}}}} \right)} \right\}}
\end{array} \right),
\end{align}
where $M$ refers to the number of Euler summation terms, $Q$ refers to the truncation order, the discretization error is bounded by $|\epsilon |\le e^{-A}/(1-e^{-A})$ and the truncation error is manageable by properly choosing $M$ and $Q$. To have a discretization error up to $10^{-10}$, $A$ is set to $A \approx 23$. As suggested by \cite{abate1995numerical}, $M=11$ and $Q=15$ are typically chosen.

By putting \eqref{eqn:lap_approx} into \eqref{eqn:out_mdash}, we can get the final expression of the conditional outage probability ${p_{k|\mathcal H}}$.
\subsubsection{Average Outage Probability}
Furthermore, the average outage probability by taking the expectation over the distance $d_{\tilde k}$ can be expressed as
\begin{equation}\label{eqn:p_tilde_def1}
{p_{\tilde k}}= {{\mathbb{E}}_{{d_{\tilde k}}}}\left\{ {p_{\tilde k|{\mathcal H}}} \right\} =1- {{\mathbb{E}}_{{d_{\tilde k}}}}\left\{ {q_{\tilde k|{\mathcal H}}} \right\},
\end{equation}
where
\begin{align}\label{eqn:q_avg_dkt}
&{\mathbb E_{{d_{\tilde k}}}}\left\{ {{q_{\tilde k|{\cal H}}}} \right\} = \frac{1}{{2\pi {\rm{i}}}}\int\limits_{\tilde c - {\rm{i}}\infty }^{\tilde c + {\rm{i}}\infty } {{e^{{{\bar \tau }_{\tilde k}}s}}} 
\underbrace {\frac{{\prod\nolimits_{i = 1}^K {{e^{ - \frac{{s{{\left| {{\zeta _{\tilde i}}} \right|}^2}}}{{1 + s{\delta _{\tilde i}}}}}}} \overbrace {{\mathbb E_{{d_{\tilde k}}}}\left\{ {{e^{ - \frac{{{{\sigma_{\tilde k} ^2}}}}{{P\ell ( {{d_{\tilde k}}})}}s - \pi {\lambda _b}{\omega _{\tilde k}}{d_{\tilde k}}^2{s^{\frac{2}{\alpha }}}}}} \right\}}^{{\varphi _{\tilde k}}\left( s \right)}}}{{s\prod\nolimits_{i = 1}^K {\left( {1 + s{\delta _{\tilde i}}} \right)} }}}_{{g_{\tilde k}}\left( s \right)}ds,
\end{align}
and ${{\bar \tau }_{\tilde k}} = \left( {{1}/{({{2^{{R_{\tilde k}}}} - 1})} - {\beta _k}^2} \right){\beta _{\tilde k}}^2{\left| {{\mu _{\tilde k}}} \right|^2}$. Likewise, 
by invoking the method of Abate and Whitt, ${p_{\tilde k}}$ is obtained as
\begin{align}\label{eqn:p_k_tildefina}
{p_{\tilde k}} \approx& 1- \frac{{{2^{ - M}}{e^{A/2}}}}{{{\bar \tau _{\tilde k}}}}\sum\limits_{m = 0}^M {{{M}\choose{m}}}
\left( \begin{array}{l}
\frac{1}{2}{\rm{Re}}\left\{ {{g_{\tilde k}}\left( {\frac{A}{{2{\bar \tau _{\tilde k}}}}} \right)} \right\} +
\sum\limits_{n = 1}^{Q + m} {{{\left( { - 1} \right)}^n}{\rm{Re}}\left\{ {{g_{\tilde k}}\left( {\frac{{A + 2n\pi {\rm{i}}}}{{2{\bar \tau _{\tilde k}}}}} \right)} \right\}}
\end{array} \right).
\end{align}
It is worth noting that the policy of user pairing directly determines the distribution of $d_{\tilde k}$. Under different grouping policies, we have different forms of ${\varphi_{\tilde k}}\left( s \right)$. In the sequel, ${\varphi_{\tilde k}}\left( s \right)$ is derived by considering two grouping policies, i.e., the random grouping and the distance-based grouping.
\begin{theorem}\label{the:rand}
Regarding the random NOMA grouping policy, $2K$ users are randomly grouped into $K$ pairs. Under this grouping policy, ${\varphi_{\tilde k}}\left( s \right)$ is given by 
\begin{align}\label{eqn:varphi_tdfin}
{\varphi _{\tilde k}}\left( s \right) =&
\frac{{2c}}{{c + {\omega _{\tilde k}}{{s}^{\frac{2}{\alpha }}}}}H_{1,1}^{1,1}\left( {\left. {\frac{{{{\sigma_{\tilde k} ^2}}}}{P}s{{\left( {\pi {\lambda _b}\left( {c + {\omega _{\tilde k}}{{s}^{\frac{2}{\alpha }}}} \right)} \right)}^{ - \frac{\alpha }{2}}}} \right|\begin{array}{*{20}{c}}
{\left( {0,\frac{\alpha }{2}} \right)}\\
{\left( {0,1} \right)}
\end{array}} \right)\notag\\
& - \frac{{2c}}{{2c + {\omega _{\tilde k}}{{s}^{\frac{2}{\alpha }}}}}H_{1,1}^{1,1}\left( {\left. {\frac{{{{\sigma_{\tilde k} ^2}}}}{P}s{{\left( {\pi {\lambda _b}\left( {2c + {\omega _{\tilde k}}{{s}^{\frac{2}{\alpha }}}} \right)} \right)}^{ - \frac{\alpha }{2}}}} \right|\begin{array}{*{20}{c}}
{\left( {0,\frac{\alpha }{2}} \right)}\\
{\left( {0,1} \right)}
\end{array}} \right),
\end{align}
where $H_{p,q}^{m,n}(\cdot)$ represents the Fox H-function \cite[Def.1.1]{mathai2009h}. Moreover, under the interference-limited regime, i.e., $\sigma^2/P\to 0$, \eqref{eqn:varphi_tdfin} reduces to
\begin{equation}\label{eqn:varphi_ob}
{\varphi _{\tilde k}}\left( s \right) = \frac{{2c}}{{c + {\omega _{\tilde k}}{{s}^{\frac{2}{\alpha }}}}} - \frac{{2c}}{{2c + {\omega _{\tilde k}}{{s}^{\frac{2}{\alpha }}}}}.
\end{equation}
\end{theorem}
\begin{proof}
Since the distance between a user and its serving BS follows a Rayleigh distribution as \eqref{eqn:pdf_d_approx}, the PDF of ${{d_{\tilde k}}}$ can be obtained by using order statistics as
\begin{equation}\label{eqn:dkt_dis}
{f_{{d_{\tilde k}}}}\left( x \right) = 2{F_d}(x){f_d}(x),
\end{equation}
where ${F_d}(x)=1 - {e^{ - c{\lambda _b}\pi {x^2}}}$ corresponds to the cumulative distribution function (CDF) of $d$. Based on \eqref{eqn:dkt_dis}, ${\varphi _{\tilde k}}\left( s \right)$ can be derived after some algebraic manipulations as
\begin{align}\label{eqn:varphitd}
{\varphi _{\tilde k}}\left( s \right) =& 2c{\lambda _b}\pi \int\limits_0^\infty  {{e^{ - \frac{{{{\sigma_{\tilde k} ^2}}}}{P}s{y^{\frac{\alpha }{2}}} - \pi {\lambda _b}\left( {{\omega _{\tilde k}}{s^{\frac{2}{\alpha }}} + c } \right)y}}dy}
-2c{\lambda _b}\pi \int\limits_0^\infty  {{e^{ - \frac{{{{\sigma_{\tilde k} ^2}}}}{P}s{y^{\frac{\alpha }{2}}} - \pi {\lambda _b}\left( {{\omega _{\tilde k}}{s^{\frac{2}{\alpha }}} + 2c } \right)y}}dy}.
\end{align}
By using the Parseval's type property of Mellin transform \cite[eq.8.3.21]{debnath2010integral} along with \cite[Eq.1.37]{mathai2009h}, \eqref{eqn:varphitd} can be represented in terms of Fox's H-function as \eqref{eqn:varphi_tdfin}. Moreover, by using the asymptotic expansion of \cite[Theorem.1.2]{mathai2009h}, the Fox's H-function in \eqref{eqn:varphi_tdfin} approaches to 1 as $\sigma^2/P\to 0$. Thus, we have \eqref{eqn:varphi_ob}.
\end{proof}

\begin{theorem}\label{the:disb}
Regarding the distance-based NOMA grouping policy, $2K$ users are assumed to rank in an increasing order according to their distances from their serving BS. Let $r_{\tilde k}$ denote the ranking order of user $\tilde k$. Under this grouping policy, ${\varphi_{\tilde k}}\left( s \right)$ is given by
\begin{align}\label{eqn:varphi_dis}
&{\varphi _{\tilde k}}\left( s \right) =\notag\\
& c{r_{\tilde k}} {2K\choose{{r_{\tilde k}}}}\sum\limits_{l = 0}^{{r_{\tilde k}} - 1} {\frac{{{{\left( { - 1} \right)}^l}{{r_{\tilde k}} - 1\choose{l}}}H_{1,1}^{1,1}\left( {\left. {\frac{{{{\sigma_{\tilde k} ^2}}}}{P}s{{\left( {\pi {\lambda _b}\left( {\left( {2K - {r_{\tilde k}} + l + 1} \right)c + {\omega _{\tilde k}}{s^{\frac{2}{\alpha }}}} \right)} \right)}^{ - \frac{\alpha }{2}}}} \right|\begin{array}{*{20}{c}}
{\left( {0,\frac{\alpha }{2}} \right)}\\
{\left( {0,1} \right)}
\end{array}} \right)}{{\left( {2K - {r_{\tilde k}} + l + 1} \right)c + {\omega _{\tilde k}}{s^{\frac{2}{\alpha }}}}}}.
\end{align}
In the interference-limited regime, i.e., $\sigma^2/P\to 0$, it follows that
\begin{equation}\label{eqn:varphi_dis_int}
{\varphi _{\tilde k}}\left( s \right) = c{r_{\tilde k}}{{2K}\choose{r_{\tilde k}}}\sum\limits_{l = 0}^{{r_{\tilde k}} - 1} {\frac{{{{\left( { - 1} \right)}^l}{{{r_{\tilde k}} - 1}\choose{l}}}}{{\left( {2K - {r_{\tilde k}} + l + 1} \right)c + {\omega _{\tilde k}}{s^{\frac{2}{\alpha }}}}}}.
\end{equation}
\end{theorem}
\begin{proof}
By using order statistics, the PDF of ${{d_{\tilde k}}}$ is expressed as
\begin{equation}\label{eqn:dkt_order}
{f_{{d_{\tilde k}}}}\left( x \right) = {r_{\tilde k}} {{2K}\choose{r_{\tilde k}}}{\left( {{F_d}(x)} \right)^{{r_{\tilde k}} - 1}}{\left( {1 - {F_d}(x)} \right)^{2K - {r_{\tilde k}}}}{f_d}(x).
\end{equation}
Then, using \eqref{eqn:dkt_order} along with some basic manipulations results in \eqref{eqn:varphi_dis}.
\end{proof}

Surprisingly, from the above two theorems, the average outage probability ${p_{\tilde k}}$ is independent of the intensity of BSs in the interference-limited regime because ${\varphi _{\tilde k}}\left( s \right)$ is irrelevant to ${{\lambda _b}}$. This is due to the fact that the effect of $\lambda_b$ has two-fold. On one hand, the increase of $\lambda_b$ yields severe co-channel interference, which eventually leads to high outage probability. On the other hand, the size of the Voronoi cell diminishes with the increase of $\lambda_b$, thus it leads to the decay of path loss.

\subsection{Outage Probability of the Near User $k$}

\subsubsection{Conditional Outage Probability}
By conducting the SIC at the near user, the near user has to subtract the far user's signal prior to decoding its own signal. Hence, the conditional outage probability of user $k$ given $\mathcal H$ is given by
\begin{align}\label{eqn:out_m}
{p_{k|\mathcal H}} &= \Pr \left( {{{\log }_2}\left( {1 + {{\rm SINR} _{k \to \tilde k}}} \right) < { R_{\tilde k}}\bigcup{{{\log }_2}\left( {1 + {{\rm SINR} _k} } \right) < { R_k}}} \right)\notag \\
 &= 1 - \underbrace{\Pr \left( {{{\rm SINR} _{k \to \tilde k}} \ge {2^{{R_{\tilde k}}}} - 1, { {{{\rm SINR} _k}} \ge {2^{{ R_k}}} - 1}} \right)}_{{q_{ k|{\mathcal{H}}}}}.
\end{align}
Clearly, ${\rm{SIN}}{{\rm{R}}_{k \to \tilde k}} $ and ${\rm{SIN}}{{\rm{R}}_k}$ are strongly correlated due to the involvement of the same co-channel interference and channel uncertainty. The existence of the spatial correlation challenges the later outage analysis. By using (\ref{eqn:gamma_sinr_mdash}) and (\ref{eqn:sinr_mm}), ${q_{ k|{\mathcal{H}}}}$ can be rewritten as 
\begin{equation}\label{eqn:out_m_subs}
{q_{ k|{\mathcal{H}}}} = \Pr \left( \begin{array}{l}
{\beta _{\tilde k}}^2{\left| {{\chi _k}} \right|^2} + {\beta _k}^2{\left| {{\mu _k} + {\chi _k}} \right|^2} + \sum\nolimits_{i = 1,i \ne k}^K {{{\left| {{\mu _i} + {\chi _i}} \right|}^2}}  \le {\theta _{\tilde k}} - \frac{{{I_k}}}{{P\ell \left( {{d_k}} \right)}},\\
{\left| {{\chi _k}} \right|^2} + \sum\nolimits_{i = 1,i \ne k}^K {{{\left| {{\mu _i} + {\chi _i}} \right|}^2}}  \le {\theta _k} - \frac{{{I_k}}}{{P\ell \left( {{d_k}} \right)}}
\end{array} \right),
\end{equation}
where ${\mu _i} = {{\bf{u}}_k}^{\rm{H}}{{{\bf{\hat H}}}_{zk}}{{\bf{v}}_i}$, ${\chi _i} = {{\bf{u}}_k}^{\rm{H}}{{\bf{E}}_{zk}}{{\bf{v}}_i}$, ${\theta _k} = {{{{\left| {{\mu _k}} \right|}^2}{\beta _k}^2}}/{({{2^{{R_k}}} - 1})} - {{{{\sigma_{k} ^2}}}}/{({P\ell \left( {{d_k}} \right)})}$ and ${\theta _{\tilde k}} = {{{{\left| {{\mu _k}} \right|}^2}{\beta _{\tilde k}}^2}}/{({{2^{{R_{\tilde k}}}} - 1})} - {{{{\sigma_{k} ^2}}}}/{({P\ell \left( {{d_k}} \right)})}$.

Similarly, it can be easily proved that ${\bs{ {\chi}}}=(\chi_1,\cdots,\chi_K)^{\rm T}= {\left( {{{\bf{u}}_{ k}}^{\rm{H}}{{\bf{E}}_{z k}}{\bf{V}}} \right)^{\rm{T}}}$ is a CSCG random vector with mean zero and covariance matrix ${\bf{ \Sigma }} = {\sigma _h^2}{{\bf{u}}_{ k}}^{\rm{H}}{{\bf{R}}_{rk}}{{\bf{u}}_{ k}}{\left( {{{\bf{V}}^{\rm{H}}}{{\bf{R}}_{tk}}{\bf{V}}} \right)^{\rm{T}}} $, i.e., ${\bs{ {\chi}}}\sim {\cal CN}(\bf 0 ,\bf\Sigma )$. With the help of the step unit function, ${q_{ k|{\mathcal{H}}}}$ can thus be expressed as 
\begin{align}\label{eqn:poutm_indi}
{q_{ k|{\mathcal{H}}}} ={{\mathbb{E}}_{{I_k}}}\left\{ \begin{array}{l}
\int\limits_{{\mathbb C^K}} {u\left( {{\theta _{\tilde k}} - \frac{{{I_k}}}{{P\ell \left( {{d_k}} \right)}} - {\beta _{\tilde k}}^2{{\left| {{\chi _k}} \right|}^2} - {\beta _k}^2{{\left| {{\mu _k} + {\chi _k}} \right|}^2} - \sum\nolimits_{i = 1,i \ne k}^K {{{\left| {{\mu _i} + {\chi _i}} \right|}^2}} } \right)} \\
 \times u\left( {{\theta _k} - \frac{{{I_k}}}{{P\ell \left( {{d_k}} \right)}} - {{\left| {{\chi _k}} \right|}^2} - \sum\nolimits_{i = 1,i \ne k}^K {{{\left| {{\mu _i} + {\chi _i}} \right|}^2}} } \right){f_{\bs\chi} }\left( {\bf{x}} \right)d{\bf{x}}
\end{array} \right\},
\end{align}
where ${f_{\bs{ \chi}}}\left( {\bf{x}} \right)$ stands for the joint PDF of ${\bs{ \chi}}$ and ${f_{\bs{ \chi}}}\left( {\bf{x}} \right) ={\exp({ - {{\bf{x}}^{\rm{H}}}{{\bf{\Sigma }}^{ - 1}}{\bf{x}}})}/{({{\pi ^K}\det  ({\bf{\Sigma }}) })}$. In analogous to \eqref{eqn:inve_Lap_out_tilde}, by using the inverse Laplace transform of the step unit function along with some rearrangements, ${q_{ k|{\mathcal{H}}}}$ can be obtained as 
\begin{equation}\label{eqn:poutm_rearrage}
{q_{ k|{\mathcal{H}}}}  = \frac{1}{{{{\left( {2\pi {\rm{i}}} \right)}^2}}}\int\limits_{c_1 - {\rm{i}}\infty }^{c_1 + {\rm{i}}\infty } {\int\limits_{c_2 - {\rm{i}}\infty }^{c_2 + {\rm{i}}\infty } {\frac{1}{s}\frac{1}{t}{e^{s{\theta _{\tilde k}} + t{\theta _k}}}{{\mathbb{E}}_{{I_k}}}\left\{ {{e^{ - \left( {s + t} \right)\frac{{{I_k}}}{{P\ell \left( {{d_k}} \right)}}}}} \right\}\mathcal K\left( {s,t} \right)dsdt} } ,\, c_1, c_2 > 0,
\end{equation}
where
\begin{align}\label{eqn:joint_integ}
&\mathcal K\left( {s,t} \right)= 
\int\limits_{{\mathbb C^M}} {{e^{ - (s{\beta _{\tilde k}}^2+t){{\left| {{x _k}} \right|}^2}  - s{\beta _k}^2{{\left| {{\mu _k} + {x _k}} \right|}^2} - \left( {s + t} \right)\sum\nolimits_{i = 1,i \ne k}^K {{{\left| {{\mu _i} + {x _i}} \right|}^2}} }}}
{f_{\bs \chi} }\left( {\bf{x}} \right)d{\bf{x}}.
\end{align}
In regard to the expectation term in \eqref{eqn:poutm_rearrage}, it can be derived similarly to \eqref{eqn:laplacefunc_faruser} as
\begin{align}\label{eqn:laplacefunc_farusernt}
{{\mathbb{E}}_{{I_{ k}}}}\left\{ {{e^{-s\frac{{{I_{ k}}}}{{P\ell \left( {{d_{ k}}} \right)}}}}} \right\}&= {e^{ - \pi {\lambda _b}{\omega _{ k}}{d_{ k}}^2{(s+t)^{\frac{2}{\alpha }}}}},
\end{align}
where ${\omega _{ k}} = \Gamma \left( {1 - \frac{2}{\alpha }} \right){\left( {\frac{{{\rho _I}}}{P}{{\left| {{{\bf{u}}_{ k}}^{\rm{H}}{{\bf{1}}_K}} \right|}^2}} \right)^{\frac{2}{\alpha }}}$. Additionally, in order to get $\mathcal K\left( {s,t} \right)$, we define the following vectors ${\bs{\mu }} = \left( \mu_1,\cdots,\mu_K \right)^{\rm T}$,
\begin{equation}\label{eqn:A_matrix}
{\bf{A}} = {\rm{diag}}\left( {\overbrace {\underbrace {s + t, \cdots ,s + t}_{\left( {k - 1} \right) - {\rm{entries}}},s{\beta _k}^2,s + t, \cdots ,s + t}^{K - {\rm{entries}}}} \right),
\end{equation}
\begin{equation}\label{eqn:B_matrix}
{\bf{B}} = {\rm{diag}}\left( {\overbrace {\underbrace {0, \cdots ,0}_{\left( {k - 1} \right) - {\rm{entries}}},s{\beta _{\tilde k}}^2 + t,0, \cdots ,0}^{K - {\rm{entries}}}} \right).
\end{equation}
Hereby, $\mathcal K\left( {s,t} \right)$ can be rewritten as
\begin{align}\label{eqn:kappa_rew}
{\cal K}(s,t) &=\frac{1}{{{\pi ^K}\det \left( {\bf{\Sigma }} \right)}}\int\nolimits_{{\mathbb C^K}} {{e^{ - {{\left( {{\bf{x}} + {\bs{\mu }}} \right)}^{\rm{H}}}{\bf{A}}\left( {{\bf{x}} + {\bs{\mu }}} \right) - {{\bf{x}}^{\rm{H}}}\left( {{\bf{B}} + {{\bf{\Sigma }}^{ - 1}}} \right){\bf{x}}}}d{\bf{x}}}.
\end{align}
By using the sum of two quadratic forms \cite[Eq.8.1.7]{petersen2008matrix}, we obtain
\begin{align}\label{eqn:identity_important}
&{\left( {{\bf{x}} + {\bs{\mu }}} \right)^{\rm{H}}}{\bf{A}}\left( {{\bf{x}} + {\bs{\mu }}} \right) + {{\bf{x}}^{\rm{H}}}\left( {{\bf{B}} + {{\bf{\Sigma }}^{ - 1}}} \right){\bf{x}} 
={\left( {{\bf{x}} + {\bs{\nu}}} \right)^{\rm{H}}}{\bf{\Xi }}\left( {{\bf{x}} + {\bs{\nu}}} \right) + \phi \left( {s,t} \right),
\end{align}
where ${\bf{\Xi }} = {\bf{A} + \bf{B}} + {{\bf{\Sigma }}^{ - 1}}=(s+t){\bf I} +{{\bf{\Sigma }}^{ - 1}}$, ${\bs{\nu }} = {{\bf{\Xi }}^{ - 1}}{\bf{A  }}\bs\mu$ and
\begin{align}\label{eqn:phi_def}
\phi\left( {s,t} \right) &= {{\bs{\mu }}^{\rm{H}}}{\bf{A}}{\bs{\mu }}-{{\bs{\mu }}^{\rm{H}}}{\bf{A}}{{\bf{\Xi }}^{ - 1}}{\bf A}{\bs{\mu }}= {{\bs{\mu }}^{\rm{H}}}{\bf{A}}{{\bf{\Xi }}^{ - 1}}{(\bf{B} + {{\bf{\Sigma }}^{ - 1}})}{\bs{\mu }}.
\end{align}
Substituting \eqref{eqn:identity_important} into \eqref{eqn:kappa_rew} yields
\begin{align}\label{eqn:kappa_rewri}
 {\cal K}(s,t)&= \frac{{{e^{ - \phi\left( {s,t} \right)}}}}{{{\pi ^K}\det \left( {\bf{\Sigma }} \right)}}\int\limits_{{\mathbb C^K}} {{e^{ - {{\left( {{\bf{x}} + {\bs{\nu }}} \right)}^{\rm{H}}}{\bf{\Xi }}\left( {{\bf{x}} + {\bs{\nu }}} \right)}}d{\bf{x}}}
 =\frac{{{e^{ - \phi \left( {s,t} \right)}}}}{{\det \left( {\bf{\Sigma }} \right)\det \left( {\bf{\Xi }} \right)}},
\end{align}
where the last step holds by using \cite[Eq. 22]{al2016distribution}.
By plugging \eqref{eqn:laplacefunc_farusernt} and \eqref{eqn:kappa_rewri} into (\ref{eqn:poutm_rearrage}), it results in
\begin{align}\label{eqn:pout_after_inner_inte}
&{q_{k|\mathcal H}} = \frac{1}{{\det \left( {\bf{\Sigma }} \right)}}\frac{1}{{{{\left( {2\pi {\rm{i}}} \right)}^2}}}
\int\limits_{c_1 - {\rm{i}}\infty }^{c_1 + {\rm{i}}\infty } {\int\limits_{c_2 - {\rm{i}}\infty }^{c_2 + {\rm{i}}\infty } {\frac{1}{st\det \left( {\bf{\Xi }} \right)}{e^{s{\theta _{\tilde k}} + t{\theta _k} - \pi {\lambda _b}{\omega _k}{d_k}^2{{(s + t)}^{\frac{2}{\alpha }}} - \phi \left( {s,t} \right)}}dsdt} }.
\end{align}

Moreover, by capitalizing on eigenvalue decomposition ${\bf{\Sigma  = \Psi \Delta }}{{\bf{\Psi }}^{\bf{H}}}$, defining ${\bf{\Psi }} = \left( {{{\bs{\psi }}_1}, \cdots ,{{\bs{\psi }}_K}} \right)$ and ${\bf{\Delta }} = {\rm{diag}}\left( {{\delta _1}, \cdots ,{\delta _K}} \right)$, $\phi \left( {s,t} \right)$ in \eqref{eqn:pout_after_inner_inte} can be simplified as
\begin{equation}\label{eqn:phi_simp}
\phi \left( {s,t} \right) = \sum\limits_{i = 1}^K {\frac{{{{\bs{\mu }}^{\rm{H}}}{\bf{A}}{{\bs{\psi }}_i}\left( {{\delta _i}{{\bs{\psi }}_i}^{\rm{H}}{{\bf B}{\bs\mu }} + {{\bs{\psi }}_i}^{\rm{H}}{\bs{\mu }}} \right)}}{{1 + \left( {s + t} \right){\delta _i}}}}.
\end{equation}
As a consequence, ${q_{k|{\mathcal H}}}$ is derived as
\begin{align}\label{eqn:q_k_fin}
{q_{k|{\mathcal H}}} &= \frac{1}{{{{\left( {2\pi {\rm{i}}} \right)}^2}}}\int\limits_{c_1 - {\rm{i}}\infty }^{c_1 + {\rm{i}}\infty } {\int\limits_{c_2 - {\rm{i}}\infty }^{c_2 + {\rm{i}}\infty } {{e^{s{\theta _{\tilde k}} + t{\theta _k}}}\underbrace {\frac{{{e^{ - \pi {\lambda _b}{\omega _k}{d_k}^2{{(s + t)}^{\frac{2}{\alpha }}} - \sum\nolimits_{i = 1}^K {\frac{{{{\bs{\mu }}^{\rm{H}}}{\bf{A}}{{\bs{\psi }}_i}\left( {{\delta _i}{{\bs{\psi }}_i}^{\rm{H}}{{\bf B}{\bs \mu }} + {{\bs{\psi }}_i}^{\rm{H}}{\bs{\mu }}} \right)}}{{1 + \left( {s + t} \right){\delta _i}}}} }}}}{{st\prod\nolimits_{i = 1}^K {\left( {1 + \left( {s + t} \right){\delta _i}} \right)} }}}_{{F_k}\left( s,t \right)}dsdt} }\notag\\
&\triangleq f_k(\theta_{\tilde k},\theta_{ k}).
\end{align}

The above two-dimensional inverse Laplace transform can be obtained by means of the Moorthy algorithm \cite{moorthy1995numerical}. More specifically, $f_k(\theta_{\tilde k},\theta_{ k})$ is approximated by adopting a trapezoidal rule of integration as 
\begin{equation}\label{eqn:f_k_app}
f_k(\theta_{\tilde k},\theta_{ k}) \approx \frac{e^{c_1\theta_{\tilde k}+c_2\theta_{ k}}}{4T^2}\left\{ {2{\rm{Re}}\left\{ \begin{array}{l}
\sum\nolimits_{{l_1} = 1}^\infty  {\sum\nolimits_{{l_2} = 1}^\infty  {F_k^{ {l_1}, {l_2}}E_{\theta_{\tilde k},\theta_{ k}}^{ {l_1}, {l_2}}} + {F_k^{ {l_1},-{l_2}}E_{\theta_{\tilde k},\theta_{ k}}^{ {l_1},-{l_2}}}}  \\
  +\sum\nolimits_{{l_1} = 1}^\infty  {F_k^{ {l_1},0}E_{\theta_{\tilde k}}^{ {l_1}}}  + \sum\nolimits_{{l_2} = 1}^\infty  {F_k^{0, {l_2}}E_{\theta_{ k}}^{ {l_2}}}
\end{array} \right\} + F_k^{0,0}} \right\},
\end{equation}
where the parameter $T$ determines the sampling period, $L$ is the truncation order,
\begin{equation}\label{eqn:F_kl1l2}
F_k^{{l_1},{l_2}} = {F_k}\left( {{c_1} + {\rm i}{l_1}{\pi}/T,{c_2} + {\rm i}{l_2}{\pi}/T} \right),
\end{equation}
\begin{equation}\label{eqn:F_En1n2}
E_{\theta_{\tilde k},\theta_{ k}}^{{l_1},{l_2}} = {e^{{\rm{i}}{l_1}\pi\theta_{\tilde k}/{T} + {\rm{i}}{l_2}\pi\theta_{ k}/{T}}} = E_{\theta_{\tilde k}}^{{l_1}}E_{\theta_{ k}}^{{l_2}},
\end{equation}
As proved in \cite{moorthy1995numerical}, the discretization error $E_r$ can be controlled by properly setting $c_i$ such that $c_2=-1/(2T)\ln((E_r-a)/(1-\xi))$, $\xi=e^{-2Tc_1}<E_r$. According to the Moorthy algorithm, \eqref{eqn:f_k_app} can be calculated by truncating the infinite series such that $l_i\in [0,L]$. The tests in \cite{moorthy1995numerical} suggest $L/\max(\theta_{\tilde k},\theta_{ k}) \in [0.5,0.8]$. Furthermore, the Epsilon algorithm is invoked to extrapolate the remaining terms above $L$ terms so as to accelerate the convergence and accuracy\cite{macdonald1964accelerated}. To be specific, $2{\cal P}+1$ partial sums are used for the Epsilon algorithm, and the experiments in \cite{brancik2004error} suggest ${\cal P}=2$.

Thereafter, by putting \eqref{eqn:q_k_fin} into \eqref{eqn:out_m}, we can obtain the conditional outage probability ${p_{k|\mathcal H}}$.
\subsubsection{Approximate Expression of ${p_{k|{\cal H}}}$}
Unfortunately, the two-dimensional numerical inversion of Laplace transform entails considerable computation burden, which precludes to the real-time optimal system design. In what follows, an approximate expression of ${p_{k|\mathcal H}}$ is derived to combat this dilemma. By overlooking the dependence between ${\rm{SIN}}{{\rm{R}}_{k \to \tilde k}} $ and ${\rm{SIN}}{{\rm{R}}_k}$, ${p_{k|{\cal H}}}$ is approximated by using \eqref{eqn:out_m} as
\begin{align}\label{eqn:pkupper}
{p_{k|{\cal H}}} \approx 1 - \underbrace {\Pr \left( {{\rm{SIN}}{{\rm{R}}_{k \to \tilde k}} \ge {2^{{R_{\tilde k}}}} - 1} \right)}_{{q_{k \to \tilde k|{\cal H}}}}
\underbrace {\Pr \left( {{\rm{SIN}}{{\rm{R}}_k} \ge {2^{{R_k}}} - 1} \right)}_{{q_{k \to  k|{\cal H}}}}.
\end{align}
Likewise, ${q_{k \to \tilde k|{\cal H}}}$ and ${q_{k \to \tilde k|{\cal H}}}$ can be obtained by following the same steps as in Section \ref{sec:outkt}. To be specific, by introducing two vectors ${{\bs{\nu }}_1} = {\left( {{\mu _1}, \cdots ,{\mu _{k - 1}},{\beta _k}^2{\mu _k},{\mu _{k + 1}}, \cdots ,{\mu _K}} \right)^{\rm{T}}}$ and ${{\bs{\nu }}_2} = {\left( {{\mu _1}, \cdots ,{\mu _{k - 1}},0,{\mu _{k + 1}}, \cdots ,{\mu _K}} \right)^{\rm{T}}}$, ${q_{k \to \tilde k|{\cal H}}}$ and ${q_{k \to  k|{\cal H}}}$ can be respectively simplified as
\begin{align}\label{eqn:pkhup_rew0}
&{q_{k \to \tilde k|{\cal H}}} = 
\Pr \left( {{{\left( {{\bs{\chi }} + {{\bs{\nu }}_1}} \right)}^{\rm{H}}}\left( {{\bs{\chi }} + {{\bs{\nu }}_1}} \right) \ge {\theta _{\tilde k}} - {\beta _k}^2{\beta _{\tilde k}}^2{{\left| {{\mu _k}} \right|}^2}} -\frac{{{I_k}}}{{P\ell \left( {{d_k}} \right)}}\right),
\end{align}
\begin{align}\label{eqn:pkhup_rew1}
{q_{k \to  k|{\cal H}}} =&\Pr \left( {{{\left( {{\bs{\chi }} + {{\bs{\nu }}_2}} \right)}^{\rm{H}}}\left( {{\bs{\chi }} + {{\bs{\nu }}_2}} \right) \ge {\theta _k}}-\frac{{{I_k}}}{{P\ell \left( {{d_k}} \right)}} \right),
\end{align}
where \eqref{eqn:pkhup_rew0} is derived in analogous to \eqref{eqn:p_out_cnd_far_simp}. Apparently, both \eqref{eqn:pkhup_rew0} and \eqref{eqn:pkhup_rew1} can be computed as \eqref{eqn:lap_approx} by replacing $({\tau _{\tilde k}},{\bs{\tilde \nu }},{\bf{\tilde \Sigma }},{\omega _{\tilde k}},{d_{\tilde k}})$ with $({\theta _{\tilde k}} - {\beta _k}^2{\beta _{\tilde k}}^2{{\left| {{\mu _k}} \right|}^2},{\bs{ \nu }}_1,{\bf{ \Sigma }},{\omega _{ k}},{d_{ k}})$ and $({\theta _k},{\bs{ \nu }}_2,{\bf{ \Sigma }},{\omega _{ k}},{d_{ k}})$, respectively. The explicit expressions of ${q_{k \to \tilde k|{\cal H}}}$ and ${q_{k \to  k|{\cal H}}}$ are omitted here to conserve space. Moreover, the approximate expression of \eqref{eqn:pkupper} will underestimate the outage performance due to the neglect of positive correlation between successful events of subtracting the far user's message and decoding its own message at user $k$. This assertion is further justified by the simulation results in Section \ref{label:num}. Hence, \eqref{eqn:pkupper} in fact serves as an upper bound of ${p_{k|{\cal H}}}$. With \eqref{eqn:pkupper}, a robust optimal design can then be enabled.
\subsubsection{Average Outage Probability}
Moreover, the average outage probability by taking the expectation with respect to the distance $d_k$ is given by
\begin{equation}\label{eqn:pK_avg_in}
{p_k} = {{\mathbb{E}}_{{d_{ k}}}}\left\{{{p_{ k|{\mathcal{H}}}}}\right\} = 1 - {{\mathbb{E}}_{{d_{ k}}}}\left\{{{q_{ k|{\mathcal{H}}}}}\right\},
\end{equation}
where ${{\mathbb{E}}_{{d_{ k}}}}\left\{{{q_{ k|{\mathcal{H}}}}}\right\}$ is obtained by 
\begin{align}\label{eqn:q_k_avg}
&{{\mathbb{E}}_{{d_{ k}}}}\left\{{{q_{ k|{\mathcal{H}}}}}\right\} =\notag\\
& \frac{1}{{{{\left( {2\pi {\rm{i}}} \right)}^2}}}\int\limits_{{c_1} - {\rm{i}}\infty }^{{c_1} + {\rm{i}}\infty } {\int\limits_{{c_2} - {\rm{i}}\infty }^{{c_2} + {\rm{i}}\infty } {{e^{s{{\bar \theta }_{\tilde k}} + t{{\bar \theta }_k}}}\underbrace {\frac{{{e^{ - \sum\nolimits_{i = 1}^K {\frac{{{\bs \mu ^{\rm{H}}}{\bf{A}}{\bs \psi _i}\left( {{\delta _i}{\bs \psi _i}^{\rm{H}}{\bf{B}}\bs \mu  + {\bs\psi _i}^{\rm{H}}\bs\mu } \right)}}{{1 + \left( {s + t} \right){\delta _i}}}} }}\overbrace {\mathbb E\left\{ {{e^{ - \frac{{{{\sigma_{k} ^2}}}}{{P\ell \left( {{d_k}} \right)}}\left( {s + t} \right) - \pi {\lambda _b}{\omega _k}{d_k}^2{{(s + t)}^{\frac{2}{\alpha }}}}}} \right\}}^{{\varphi _k}\left( {s + t} \right)}}}{{st\prod\nolimits_{i = 1}^K {\left( {1 + \left( {s + t} \right){\delta _i}} \right)} }}}_{{g_k}\left( {s,t} \right)}dsdt} }.
\end{align}
where ${{\bar \theta }_k} = {\left| {{\mu _k}} \right|^2}{\beta _k}^2/({2^{{R_k}}} - 1)$ and ${{\bar \theta }_{\tilde k}} = {\left| {{\mu _k}} \right|^2}{\beta _{\tilde k}}^2/({2^{{R_{\tilde k}}}} - 1)$.  By using the Moorthy algorithm, ${p_k}$ can be consequently derived as 
\begin{align}\label{eqn:pK_fina}
{p_k} 
&\approx 1 - \frac{e^{c_1\bar \theta_{\tilde k}+c_2\bar \theta_{ k}}}{4T^2}\left\{ {2{\rm{Re}}\left\{ {\begin{array}{*{20}{l}}
{\sum\nolimits_{{l_1} = 1}^L  {\sum\nolimits_{{l_2} = 1}^L  {g_k^{ {l_1}, {l_2}}E_{\bar \theta_{\tilde k},\bar \theta_{ k}}^{ {l_1}, {l_2}}} }  + {g_k^{ {l_1},-{l_2}}E_{\bar \theta_{\tilde k},\bar \theta_{ k}}^{ {l_1},-{l_2}}}}\\
{ +\sum\nolimits_{{l_1} = 1}^L  {g_k^{ {l_1},0}E_{\bar \theta_{\tilde k}}^{ {l_1}}}  + \sum\nolimits_{{l_2} = 1}^L  {g_k^{0,  {l_2}}E_{\bar \theta_{ k}}^{ {l_2}}} }
\end{array}} \right\} + g_k^{0,0}} \right\},
\end{align}
where $g_k^{{l_1},{l_2}} = {g_k}\left( {{c_1} + {\rm i}{l_1}{\pi}/T,{c_2} + {\rm i}{l_2}{\pi}/T} \right)$.

Similarly to Theorem \ref{the:rand} and Theorem \ref{the:disb}, ${\varphi _k}\left( {s + t} \right)$ are obtained by considering the random and the distance-based policies.
\begin{theorem}
Under the random NOMA grouping policy, ${\varphi_{ k}}\left( s + t\right)$ is given by
\begin{align}\label{eqn:varphi_k_rand}
&{\varphi _{ k}}\left( s+t \right) = \frac{{2c}}{{2c + {\omega _{ k}}{(s+t)^{\frac{2}{\alpha }}}}}
H_{1,1}^{1,1}\left( {\left. {\frac{{{{\sigma_{ k} ^2}}}}{P}(s+t){{\left( {\pi{\lambda _b} \left( {2c + {\omega _{ k}}{(s+t)^{\frac{2}{\alpha }}}} \right)} \right)}^{ - \frac{\alpha }{2}}}} \right|\begin{array}{*{20}{c}}
{\left( {0,\frac{\alpha }{2}} \right)}\\
{\left( {0,1} \right)}
\end{array}} \right),
\end{align}
In the interference-limited regime, i.e., $\sigma^2/P\to 0$, ${\varphi _{ k}}\left( s \right)$ can be simplified as
\begin{equation}\label{eqn:varphi_obk}
{\varphi _{ k}}\left( s+t \right) = \frac{{2c}}{{2c + {\omega _{ k}}{(s+t)^{\frac{2}{\alpha }}}}}.
\end{equation}
\end{theorem}
\begin{proof}
Similarly to \eqref{eqn:dkt_dis}, the PDF of ${{d_{ k}}}$ is given by
\begin{equation}\label{eqn:dk_dis}
{f_{{d_{ k}}}}\left( x \right) = 2\left( {1 - {F_d}(x)} \right){f_d}(x),
\end{equation}
By using \eqref{eqn:dk_dis}, \eqref{eqn:varphi_k_rand} can be obtained accordingly.
\end{proof}

\begin{theorem}\label{the:dis_k_t}
Under the distance-based NOMA grouping policy, ${\varphi_{ k}}\left( s + t \right)$ can be obtained similarly to \eqref{eqn:varphi_dis} by directly replacing $\tilde k$, $\sigma_{\tilde k}^2$, $r_{\tilde k}$ and $s$ with $ k$, $\sigma_{ k}^2$, $r_{ k}$ and $(s+t)$, respectively. Thereon, in the interference-limited regime, i.e., $\sigma^2/P\to 0$, ${\varphi _{ k}}\left( s \right)$ collapses to
\begin{align}\label{eqn:varphi_disd}
&{\varphi_{ k}}\left( s + t \right) = 
c{r_{ k}}{{2K}\choose{r_{ k}}}\sum\limits_{l = 0}^{{r_{ k}} - 1} {\frac{{{{\left( { - 1} \right)}^l}{{{r_{ k}} - 1}\choose{l}}}}{{\left( {2K - {r_{ k}} + l + 1} \right)c + {\omega _{ k}}{(s+t)^{\frac{2}{\alpha }}}}}},
\end{align}
where $r_{ k}$ denotes the ranking order of user $ k$ and $r_{ k} < r_{\tilde k}$.
\end{theorem}
\begin{proof}
The proof is parallel to that of Theorem \ref{the:disb} and hence omitted.
\end{proof}
Clearly, the average outage probability ${p_{ k}}$ is independent of the intensity of BSs because ${\varphi _{ k}}\left( s \right)$ is irrelevant to ${{\lambda _b}}$. 

\section{Asymptotic Outage Analysis}\label{sec:asy}
Unfortunately, the exact outage expressions are too cumbersome to gain extra helpful insights. To address this issue, this section seeks to investigate the asymptotic behavior of the outage probability as the quality of CSI improves, i.e., ${\mathcal K}_{zk} \to \infty$ or $\sigma_h^2 \to 0$. For simplicity and tractability, we specialize our analysis to the case of no inter-cell interference, i.e., $\lambda_b=0$. Similarly to \cite{park2012outage}, the tool of Chernoff bound is resorted to offer upper bounds for the outage probabilities. The asymptotic outage probabilities of the far and near users are individually discussed by splitting into two subsections.
\subsection{Asymptotic Outage Probability of the Far User $\tilde k$}
If $\lambda_b=0$, the conditional outage probability of user $\tilde k$ is
\begin{align}\label{eqn:out_mdash1}
{p_{\tilde k|\mathcal H}} &= \Pr \left\{ {{{\left( {{\bs{\tilde {\chi}}} + {\bs{\tilde \nu }}} \right)}^{\rm{H}}}\left( {{\bs{\tilde {\chi}}} + {\bs{\tilde \nu }}} \right) > {\tau _{\tilde k}} } \right\},
\end{align}
By using Chernoff bound, ${p_{\tilde k|\mathcal H}}$ is upper bounded by
\begin{equation}\label{eqn:chernb_kt}
{p_{\tilde k|{\cal H}}} \le {e^{ - {\tau _{\tilde k}}\tilde s}}\mathbb E\left\{ {{e^{\tilde s{{\left( {{\bs{\tilde {\chi}}} + {\bs{\tilde \nu }}} \right)}^{\rm{H}}}\left( {{\bs{\tilde {\chi}}} + {\bs{\tilde \nu }}} \right)}}} \right\},
\end{equation}
for any $\tilde s>0$. With the joint PDF of ${\bs{\tilde \chi}}$ and the important integral \eqref{eqn:identi_gassian_inte}, \eqref{eqn:chernb_kt} can be rewritten as
\begin{align}\label{eqn:pkh_rew}
{p_{\tilde k|{\cal H}}} \le \frac{{{e^{ - {\tau _{\tilde k}}\tilde s}}}}{{{\pi ^K}\det \left( {{\bf{\tilde \Sigma }}} \right)}}\int\limits_{{\mathbb C^K}} {{e^{\tilde s{{\left( {{\bf{x}} + {\bs{\tilde \nu }}} \right)}^{\rm{H}}}\left( {{\bf{x}} + {\bs{\tilde \nu }}} \right) - {{\bf{x}}^{\rm{H}}}{{{\bs{\tilde \Sigma }}}^{ - 1}}{\bf{x}}}}} d{\bf{x}}
 &= \frac{{{e^{ - {\tau _{\tilde k}}\tilde s}}}}{{\det \left( {{\bf{I}} - \tilde s{\bf{\tilde \Sigma }}} \right)}}{e^{{{{\bf{\tilde \nu }}}^{\rm{H}}}{{\left( {\frac{1}{\tilde s}{\bf{I}} - {\bf{\tilde \Sigma }}} \right)}^{ - 1}}{\bf{\tilde \nu }}}}\notag\\
 &= {e^{ - \left( {{\tau _{\tilde k}}\tilde s + \sum\limits_{i = 1}^K {\ln \left( {1 - \tilde s{\delta _{\tilde i}}} \right)}  + \sum\limits_{i = 1}^K {\frac{{\tilde s{{\left| {{\zeta _{\tilde i}}} \right|}^2}}}{{\tilde s{\delta _{\tilde i}} - 1}}} } \right)}},
\end{align}
where $\tilde s\in (0,\min\{{{1/\delta _{\tilde i}}},i\in[1,K]\})$. To circumvent the effect of ${\sigma _h^2}$ on $\tilde s$, we define ${\delta _{\tilde i}} = {\sigma _h^2}{\upsilon _{\tilde i}}$, where ${\delta _{\tilde i}}$ are the eigenvalues of ${\bs{\tilde \Sigma }} = {\sigma _h^2}{{\bf{u}}_{\tilde k}}^{\rm{H}}{{\bf{R}}_{r\tilde k}}{{\bf{u}}_{\tilde k}}{\left( {{{\bf{V}}^{\rm{H}}}{{\bf{R}}_{t\tilde k}}{\bf{V}}} \right)^{\rm{T}}}$. Accordingly, $\{\upsilon _{\tilde 1},\cdots,\upsilon _{\tilde K}\}$ are the eigenvalues of ${{\bf{u}}_{\tilde k}}^{\rm{H}}{{\bf{R}}_{r\tilde k}}{{\bf{u}}_{\tilde k}}{\left( {{{\bf{V}}^{\rm{H}}}{{\bf{R}}_{t\tilde k}}{\bf{V}}} \right)^{\rm{T}}}$. Hereby, \eqref{eqn:pkh_rew} can be expressed as
\begin{align}\label{eqn:pkh_rew1}
{p_{\tilde k|{\cal H}}} &\le {e^{ - \frac{{\bar s}}{{\sigma _h^2}} {\left( {{\tau _{\tilde k}} + \sum\limits_{i = 1}^K {\frac{{{{\left| {{\zeta _{\tilde i}}} \right|}^2}}}{{\bar s{\upsilon _{\tilde i}} - 1}}} } \right)}}} {{e^{ - \sum\limits_{i = 1}^K {\ln \left( {1 - \bar s{\upsilon _{\tilde i}}} \right)} }}},
\end{align}
where $\bar s = {\sigma _h^2}\tilde s \in (0,\min \{ 1/{\upsilon _{\tilde i}},i \in [1,K]\} )$. Clearly from \eqref{eqn:pkh_rew1}, the following theorem concerning the asymptotic behavior of ${p_{\tilde k|{\cal H}}}$ can be found.
\begin{theorem}\label{the:asyt}
As the channel uncertainty vanishes, i.e., ${\mathcal K}_{zk} \to \infty$ or ${{\sigma _h^2}} \to 0$, the outage probability ${p_{\tilde k|{\cal H}}}$ decays to zero if the target transmission rate satisfies the following sufficient condition
\begin{equation}\label{eqn:rtk_thre}
{R_{\tilde k}} < {\log _2}\left( {1 + \frac{{{\beta _{\tilde k}}^2{{\left| {{\mu _{\tilde k}}} \right|}^2}}}{{\sum\nolimits_{i = 1}^K {{{\left| {{\zeta _{\tilde i}}} \right|}^2}}  + {\beta _k}^2{\beta _{\tilde k}}^2{{\left| {{\mu _{\tilde k}}} \right|}^2}}+ \frac{{\sigma _{\tilde k}^2}}{{P\ell \left( {{d_{\tilde k}}} \right)}} }} \right).
\end{equation}
\end{theorem}
\begin{proof}
As ${{\sigma _h^2}} \to 0$, the right-hand side (RHS) of the inequality \eqref{eqn:pkh_rew1} approaches to zero only if ${{\tau _{\tilde k}} + \sum\nolimits_{i = 1}^K {{{{{\left| {{\zeta _{\tilde i}}} \right|}^2}}}/{({\bar s{\upsilon _{\tilde i}} - 1})}} }  >0$. %
By using the definition of ${\tau _{\tilde k}}$ together with some rearrangements, setting $\bar s=0$ yields the upper bound of ${R_{\tilde k}}$.
\end{proof}
It should be highlighted that Theorem \ref{the:asyt} differs from the results in \cite{park2012outage}, where only a single data stream is considered. Besides, a loose upper bound of ${R_{\tilde k}}$ is provided in \cite{park2012outage}. Particularly, the simulation results in Section \ref{label:num} justify the correctness of Theorem \ref{the:asyt}. However, it is worth mentioning that the same property does not carry over to the case of $\lambda_b\ne0$, that is, an arbitrary low outage target cannot be warranted for any nonzero target rates.

\subsection{Asymptotic Outage Probability of the Near User $k$}
From \eqref{eqn:out_m}, ${p_{k|{\cal H}}}$ is upper bounded by using inclusion-exclusion principle as
\begin{align}\label{eqn:plcalHup}
{p_{k|{\cal H}}} \le& \Pr \left( {{{\log }_2}\left( {1 + {\rm{SIN}}{{\rm{R}}_{k \to \tilde k}}} \right) < {R_{\tilde k}}} \right)
 + \Pr \left( {{{\log }_2}\left( {1 + {\rm{SIN}}{{\rm{R}}_k}} \right) < {R_k}} \right).
\end{align}
By setting $\lambda_b=0$ together with the definitions ${{\bs{\nu }}_1} = ( {\mu _1}, \cdots ,{\mu _{k - 1}},{\beta _k}^2{\mu _k},{\mu _{k + 1}}, \cdots ,{\mu _K})^{\rm{T}}$ ${{\bs{\nu }}_2} = ( {\mu _1}, \cdots ,{\mu _{k - 1}},0,{\mu _{k + 1}}, \cdots ,{\mu _K} )^{\rm{T}}$, \eqref{eqn:plcalHup} can be rewritten as
\begin{align}\label{eqn:pkhup_rew}
{p_{k|{\cal H}}} \le& \Pr \left( {{{\left( {{\bs{\chi }} + {{\bs{\nu }}_1}} \right)}^{\rm{H}}}\left( {{\bs{\chi }} + {{\bs{\nu }}_1}} \right) > {\theta _{\tilde k}} - {\beta _k}^2{\beta _{\tilde k}}^2{{\left| {{\mu _k}} \right|}^2}} \right)
 + \Pr \left( {{{\left( {{\bs{\chi }} + {{\bs{\nu }}_2}} \right)}^{\rm{H}}}\left( {{\bs{\chi }} + {{\bs{\nu }}_2}} \right) > {\theta _k}} \right),
\end{align}
where the first term on the RHS of the inequality is obtained similarly to \eqref{eqn:p_out_cnd_far_simp}.  In analogous to \eqref{eqn:chernb_kt} and \eqref{eqn:pkh_rew}, applying Chernoff bound to \eqref{eqn:pkhup_rew} gives rise to
\begin{align}\label{eqn:pkh_simp}
{p_{k|{\cal H}}} <& {e^{ - \left( {\left( {{\theta _{\tilde k}} - {\beta _k}^2{\beta _{\tilde k}}^2{{\left| {{\mu _k}} \right|}^2}} \right){s} + \sum\limits_{i = 1}^K {\ln \left( {1 - {s}{\delta _i}} \right)}  + \sum\limits_{i = 1}^K {\frac{{{s}{{\left| {{{\bs{\psi }}_i}{{\bs{\nu }}_1}} \right|}^2}}}{{s{\delta _i} - 1}}} } \right)}} 
+ {e^{ - \left( {{\theta _k}{s} + \sum\limits_{i = 1}^K {\ln \left( {1 - {s}{\delta _i}} \right)}  + \sum\limits_{i = 1}^K {\frac{{{s}{{\left| {{{\bs{\psi }}_i}{{\bs{\nu }}_2}} \right|}^2}}}{{s{\delta _i} - 1}}} } \right)}},
\end{align}
where $s \in (0,\min\{{{1/\delta _{ i}}},i\in[1,K]\})$.

Similarly to \eqref{eqn:pkh_rew1}, we denote the eigenvalues of ${{\bf{u}}_{ k}}^{\rm{H}}{{\bf{R}}_{r k}}{{\bf{u}}_{ k}}{\left( {{{\bf{V}}^{\rm{H}}}{{\bf{R}}_{t k}}{\bf{V}}} \right)^{\rm{T}}}$ by $\{\upsilon _{ 1},\cdots,\upsilon _{ K}\}$. Thus, we have ${\delta _{ i}} = {\sigma _h^2}{\upsilon _{ i}}$, where ${\delta _{ i}}$ are the eigenvalues of ${\bf{ \Sigma }}$. \eqref{eqn:pkh_simp} can thus be rewritten as
\begin{align}\label{eqn:pkh_bdfinaexp}
{p_{k|{\cal H}}} < {e^{ - \sum\limits_{i = 1}^K {\ln \left( {1 - \hat s{\upsilon _i}} \right)} }}\left( {e^{ - \frac{{\hat s}}{{\sigma _h^2}}\left( {{\theta _{\tilde k}} - {\beta _k}^2{\beta _{\tilde k}}^2{{\left| {{\mu _k}} \right|}^2}+ \sum\limits_{i = 1}^K {\frac{{{{\left| {{{\bs{\psi }}_i}^{\rm H}{{\bs{\nu }}_1}} \right|}^2}}}{{\hat s{\upsilon _i} - 1}}} } \right)}} \right. 
 \left. + {e^{ - \frac{{\hat s}}{{\sigma _h^2}}\left( {{\theta _k} + \sum\limits_{i = 1}^K {\frac{{{{\left| {{{\bs{\psi }}_i}^{\rm H}{{\bs{\nu }}_2}} \right|}^2}}}{{\hat s{\upsilon _i} - 1}}} } \right)}} \right),
\end{align}
where $\hat s = {\sigma _h^2} s \in (0,\min \{ 1/{\upsilon _{ i}},i \in [1,K]\} )$. On the basis of \eqref{eqn:pkh_bdfinaexp}, the following theorem regarding the asymptotic behavior of ${p_{k|{\cal H}}}$ can be proved.
\begin{theorem}\label{the:upp2}
As the channel uncertainty vanishes, i.e., ${\mathcal K}_{zk} \to \infty$ or ${{\sigma _h^2}} \to 0$, the outage probability ${p_{k|{\cal H}}}$ converges to zero if the target transmission rates satisfy the following sufficient conditions
\begin{equation}\label{eqn:rtk_thre1}
{R_{\tilde k}} < {\log _2}\left( {1 + \frac{{{\beta _{\tilde k}}^2{{\left| {{\mu _k}} \right|}^2}}}{{\sum\nolimits_{i = 1}^K {{{\left| {{\bs \psi _i}^{\rm H}{\bs \nu _1}} \right|}^2}}   + {\beta _k}^2{\beta _{\tilde k}}^2{{\left| {{\mu _k}} \right|}^2}}+ \frac{{\sigma _k^2}}{{P\ell \left( {{d_k}} \right)}}}} \right),
\end{equation}
\begin{equation}\label{eqn:rk_bound}
{R_k} < {\log _2}\left( {1 + \frac{{{{\left| {{\mu _k}} \right|}^2}{\beta _k}^2}}{{\sum\nolimits_{i = 1}^K {{{\left| {{{\bs{\psi }}_i}^{\rm H}{{\bs{\nu }}_2}} \right|}^2}} +\frac{{\sigma _k^2}}{{P\ell \left( {{d_k}} \right)}}  }}} \right).
\end{equation}
\end{theorem}
\begin{proof}
As ${{\sigma _h^2}} \to 0$, the RHS of the inequality \eqref{eqn:pkh_bdfinaexp} approaches to zero only if ${\theta _{\tilde k}} - {\beta _k}^2{\beta _{\tilde k}}^2{\left| {{\mu _k}} \right|^2} + \sum\nolimits_{i = 1}^K {{{{{\left| {{{\bf{\psi }}_i}{{\bf{\nu }}_1}} \right|}^2}}}/{({\hat s{\upsilon _i} - 1})}}  > 0$ and ${\theta _k} + \sum\nolimits_{i = 1}^K {{{{{\left| {{{\bf{\psi }}_i}{{\bf{\nu }}_2}} \right|}^2}}}/{({\hat s{\upsilon _i} - 1})}}  > 0$. By using the definitions of ${\theta _{\tilde k}}$ and ${\theta _{ k}}$, the upper bounds of ${R_{\tilde k}}$ and ${R_{ k}}$ can be obtained by setting $\hat s=0$.
\end{proof}
\section{Optimal System Design}\label{sec:opt}
By making use of the foregoing analytical results, the parameters of MIMO-NOMA enhanced SCNs can be properly configured to adapt to the imperfect fading channels while ensuring the quality of service, including the precoding matrix (i.e., $\bf V$), receiver filters (i.e., ${\bf u}_k, {\bf u}_{\tilde k}, k\in[1,K]$), power allocation coefficients (i.e., $\beta_k$, $\beta_{\tilde k}, k\in[1,K]$) and the transmission rates (i.e., $R_k$, $R_{\tilde k}, k\in[1,K]$). To exemplify this, we confine our attention to maximize the long-term goodput that measures the number of successfully conveyed information bits per transmission\cite{shi2019zero,shi2018energy,rui2008combined}. Given the known channel information $\mathcal H $, the conditional goodput of the MIMO-NOMA system is expressed as
\begin{align}\label{eqn:T_g_def}
{\mathcal T_{g|{\mathcal H}}} &= \sum\limits_{k = 1}^K {{R_k}( {1 - {p_{k|\mathcal H}}} ) + {R_{\tilde k}}( {1 - {p_{\tilde k|\mathcal H}}} )}. 
\end{align}
Thereon, the maximization of goodput while guaranteeing low outage constraints is posed by
\begin{equation}\label{eqn:prob_def}
\begin{array}{*{20}{c l}}
{\mathop {\max }\limits_{{\bf{V}},\{{{\bf{u}}_k},{{\bf{u}}_{\tilde k}},{R_k},{R_{\tilde k}}:k\in[1,K]\}} }&{{\mathcal T_{g|{\cal H}}}}\\
{{\rm{s}}{\rm{.t}}{\rm{.}}}&{{p_{k|{\cal H}}} \le \varepsilon ,{p_{\tilde k|{\cal H}}} \le \varepsilon ,k \in [1,K]}\\
{}&{\left\| {{{\bf{v}}_k}} \right\| = 1,k \in [1,K]}\\
{}&{{{\beta _k}^2}\left( {{2^{{R_{\tilde k}}}} - 1} \right) < 1,k \in [1,K]}
\end{array},
\end{equation}
where $\varepsilon$ denotes the maximum endurable outage probability. In general, the power allocation coefficients $\beta_k$ and transmission rates $R_k$, $R_{\tilde k}$ cannot be jointly optimized in \eqref{eqn:prob_def} in order to not violate the intention of NOMA principle. This is because as disclosed by \cite{shi2018cooperative}, the joint optimization of $\beta_k$ and $R_k$\&$R_{\tilde k}$ would behave like waterfilling algorithm if the user fairness is disregarded, that is, all the power would be allocated to the users under benign channel conditions and no information bits would be delivered to the users experienced bad channel conditions. In light of this observation, we fix the power allocation coefficient to maintain the user fairness in the sequel.

Unfortunately, due to the complex outage expressions, it is virtually impossible to get the globally optimal solution to \eqref{eqn:prob_def}. To tackle this difficulty, the concept of the signal/interference alignment proposed by \cite{ding2016general} is applied to choose an appropriate ${\bf{V}}$, ${{\bf{u}}_k}$ and ${{\bf{u}}_{\tilde k}}$ with low implementation overhead\footnote{Intrinsically, the signal alignment can be regarded as a linear decorrelator \cite[Sec. 8.3.1]{tse2005fundamentals}, which is a combination of the projection operation and matched filter. The linear decorrelator is frequently used to suppress the inter-stream interference.}. In particular, in order to substantially suppress the inter-pair interference, it is necessary to impose the constraints on
${\mu _{i}}={\mu _{\tilde i}}=0,i\ne k$. This leads to
\begin{equation}\label{eqn:inter_satis}
\left( {\begin{array}{*{20}{c}}
{{{\bf{u}}_{ k}}^{\rm{H}}{{{\bf{\hat H}}}_{z k}}}\\
{{{\bf{u}}_{\tilde k}}^{\rm{H}}{{{\bf{\hat H}}}_{z\tilde k}}}
\end{array}} \right){{\bf{v}}_i}={\bf 0},\, i \ne k, k\in [1,K],
\end{equation}
In order to capitalize on the signal alignment, we apply the following decomposition ${\bf V}={\bf L}{\bf P}$, where ${\bf L}\in{\mathbb C}^{M\times K}$, ${\bf L}^{\rm H}{\bf L}={\bf I}_K$ and ${\bf P}\in{\mathbb C}^{K\times K}$. In some sense, the involvement of ${\bf L}$ is equivalent to choose $K$ transmit antennas among $M$ ones. Hence, a similar algorithm as suggested in \cite{ding2016general} can be developed to obtain ${\bf L}$, as shown in Algorithm \ref{alg:rs}. Furthermore, by defining ${{\bf{g}}_k}^{\rm{H}} = {{\bf{u}}_k}^{\rm{H}}{{{\bf{\hat H}}}_{zk}}{\bf{L}}$ and ${{\bf{g}}_{\tilde k}}^{\rm{H}} = {{\bf{u}}_{\tilde k}}^{\rm{H}}{{{\bf{\hat H}}}_{z\tilde k}}{\bf{L}}$, we have ${\left( {{{\bf{g}}_k},{{\bf{g}}_{\tilde k}}} \right)^{\rm{H}}}{{\bf{p}}_i} = {\bf{0}}$, where ${\bf P} = ({{\bf{p}}_1},\cdots,{{\bf{p}}_K})$.
To ensure the existence of ${{\bf{p}}_i}$, the prerequisite of undertaking the signal alignment is ${{\bf{g}}_k}={{\bf{g}}_{\tilde k}}$, which ensures the existence of both ${{\bf{u}}_k}$ and ${{\bf{u}}_{\tilde k}}$. 
As proved by \cite{ding2016general}, ${{{\bf{u}}_k}}$ and ${{{\bf{u}}_{\tilde k}}}$ are thus given by
\begin{equation}\label{eqn:u_expli}
\left( {\begin{array}{*{20}{c}}
{{{\bf{u}}_k}}\\
{{{\bf{u}}_{\tilde k}}}
\end{array}} \right)={\bf U}_k {\bf z}_k,
\end{equation}
where ${\bf U}_k \in {\mathbb C}^{2N\times (2N-K)}$ consists of the $(2N-K)$ right singular vectors of $({{{( {{{{\bf{\hat H}}}_{zk}}{\bf{L}}} )}^{\rm{H}}}, - {{( {{{{\bf{\hat H}}}_{z\tilde k}}{\bf{L}}} )}^{\rm{H}}}})$ corresponding to its zero singular values, ${\bf z}_k \in {\mathbb C}^{(2N-K)\times 1}$ and $|{\bf z}_k|^2=1$. In \cite{ding2016general}, an algorithm with low complexity has been provided to properly choose ${\bf z}_k$. After determining ${{{\bf{u}}_k}}$ and ${{{\bf{u}}_{\tilde k}}}$, ${{\bf{g}}_{ k}}$ and ${{\bf{g}}_{\tilde k}}$ are obtained. Accordingly, ${\bf P}$ is then given by ${\bf{P}} = {{\bf{G}}^{ - {\rm{H}}}}{\bf{D}}$,
where ${\bf{G}} = ({\bf g}_1,\cdots,{\bf g}_K)$ and ${\bf{D}}={\rm diag}(1/\sqrt{({{\bf{G}}^{ - 1}}{{\bf{G}}^{ - {\rm{H}}}})_{1,1}},\cdots,1/\sqrt{({{\bf{G}}^{ - 1}}{{\bf{G}}^{ - {\rm{H}}}})_{K,K}})$ ensures $\left\| {{{\bf{v}}_k}} \right\| = \left\| {{{\bf{p}}_k}} \right\| = 1$. Therefore, ${\bf V}$ can be calculated as
\begin{equation}\label{eqn:V_fina}
{\bf V}={\bf L}{{\bf{G}}^{ - {\rm{H}}}}{\bf{D}}.
\end{equation}
An algorithm for the design of ${\bf V}$ is listed in Algorithm \ref{alg:rs}.
\begin{algorithm}
   \caption{The selection of the precoding matrix ${\bf{V}}$}\label{alg:rs}
    \begin{algorithmic}[1]
        \For{$i=1$ to $M!/(M-K)!$}
        \State ${\bf L}\gets {\bf Y}_i$, ${\bf Y}_i$ is a matrix of size ${M\times K}$ that contains exactly one entry of 1 in each \par column, at most one entry of 1 in each row and 0s elsewhere;
        \State Find ${{\bf{u}}_k},{{\bf{u}}_{\tilde k}}$ by using \cite[Algorithm 1]{ding2016general};
        \State  Construct ${\bf{G}} = ({\bf g}_1,\cdots,{\bf g}_K)^{\rm{H}}$,
         where ${{\bf{g}}_k}^{\rm{H}} = {{\bf{u}}_k}^{\rm{H}}{{{\bf{\hat H}}}_{zk}}{\bf{L}}$ and $k\in [1,K]$;
        \vspace{0.1em}
        \State Get the effective channel gain for each user pair, $\gamma_{k} = 1/({{\bf{G}}^{ - 1}}{{\bf{G}}^{ - {\rm{H}}}})_{k,k}$;
        \State Find the smallest effective channel gain, i.e., $\gamma_{\min,i}=\min\{\gamma_{1},\cdots,\gamma_{K}\}$;
      \EndFor
      \State Find the index $i$ that maximizes the smallest effective channel gain, i.e., \par $i^* = \mathop {\arg }\limits_{i\in[1,M!/(M-K)!]}  \max  \gamma_{\min,i}$;
      \State ${\bf V} \gets {\bf L}{{\bf{G}}^{ - {\rm{H}}}}{\bf{D}}$, where ${\bf L}={\bf Y}_{i^*}$.
        \end{algorithmic}
\end{algorithm}

After deciding ${\bf{V}}$, ${{\bf{u}}_k}$ and ${{\bf{u}}_{\tilde k}}$, the original problem of \eqref{eqn:prob_def} collapses to
\begin{equation}\label{eqn:prob_defsimp}
\begin{array}{*{20}{c l}}
{\mathop {\max }\limits_{\{{R_k},{R_{\tilde k}}:k\in[1,K]\}} }&{ \sum\limits_{k = 1}^K {{R_k}( {1 - {p_{k|\mathcal H}}} ) + {R_{\tilde k}}( {1 - {p_{\tilde k|\mathcal H}}} )} }\\
{{\rm{s}}{\rm{.t}}{\rm{.}}}&{{p_{k|{\cal H}}} \le \varepsilon ,{p_{\tilde k|{\cal H}}} \le \varepsilon ,k \in [1,K]}\\
{}&{{{\beta _k}^2}\left( {{2^{{R_{\tilde k}}}} - 1} \right) < 1,k \in [1,K]}
\end{array}.
\end{equation}
By noticing that the objective function of \eqref{eqn:prob_defsimp} is a sum of ${{{R_k}( {1 - {p_{k|\mathcal H}}} ) + {R_{\tilde k}}( {1 - {p_{\tilde k|\mathcal H}}} )}}$ for all $k\in [1,K]$, ${p_{k|{\cal H}}}$ and ${p_{\tilde k|{\cal H}}}$ only depend on ${R_k}$ and ${R_{\tilde k}}$ (independent of ${R_i}$ and ${R_{\tilde i}}$ for $i\ne k$), \eqref{eqn:prob_defsimp} can further be decoupled as $K$ independent subproblems, i.e.,
\begin{equation}\label{eqn:prob_defsimpdec}
\begin{array}{*{20}{c l}}
{\mathop {\max }\limits_{{R_k},{R_{\tilde k}}} }&{{{R_k}( {1 - {p_{k|\mathcal H}}} ) + {R_{\tilde k}}( {1 - {p_{\tilde k|\mathcal H}}} )}}\\
{{\rm{s}}{\rm{.t}}{\rm{.}}}&{{p_{k|{\cal H}}} \le \varepsilon ,{p_{\tilde k|{\cal H}}} \le \varepsilon}\\
{}&{{{\beta _k}^2}\left( {{2^{{R_{\tilde k}}}} - 1} \right) < 1}
\end{array}, k\in[1,K].
\end{equation}
\eqref{eqn:prob_defsimpdec} can be solved with many popular programming tools, such as sequential convex programming (SCP), interior-point method, etc.\footnote{To examine the computational complexity of \eqref{eqn:prob_defsimpdec}, the interior-point method is taken as an example. Since there are only two variables to be optimized in each subproblem, the computational complexity of each subproblem is mainly affected by the number of Newton iterations. According to the complexity analysis of the interior-point method in \cite{boyd2004convex}, the total number of Newton iterations is linearly proportional to $\left\lceil {\frac{{\ln \left( {m/\left( {{t^{\left( 0 \right)}\epsilon}} \right)} \right)}}{{\ln \mu }}} \right\rceil $, where $m$ is the number of inequality constraints, $\mu $ is the increasing factor, $t^{(0)}$ is the initial value of the parameter that sets the accuracy of the barrier approximation, $\epsilon$ specifies the final duality gap and $\left\lceil \cdot \right\rceil$ denotes the ceil operation. Accordingly, the computational complexity of solving the $K$ independent subproblems is linearly proportional to $K\left\lceil {\frac{{\ln \left( {m/\left( {{t^{\left( 0 \right)}\epsilon}} \right)} \right)}}{{\ln \mu }}} \right\rceil $. In fact, ${\ln \left( {m/\left( {{t^{\left( 0 \right)}\epsilon}} \right)} \right)}$ refers to the logarithm of the ratio of the initial duality gap $m/t^{\left( 0 \right)}$ to the final duality gap $\epsilon$. In practice, the factor $\mu $ is usually assigned with a larger value around 2 to 100, which leads to the total number of iterations on the order of a few tens. As a consequence, the computational complexity of solving \eqref{eqn:prob_defsimpdec} is very low.}.

\section{Numerical Analysis}\label{label:num}
In this section, the simulation results are presented for verifications and discussions. Unless otherwise specified, the simulation parameters are tabulated in Table \ref{tab:list_symb}\footnote{The intention behind the linear relationship between $R_k$ and $R_{\tilde k}$ is to plot 2-D figures for the purpose of clear illustrations.}, where an exponential correlation profile is used to model the covariance matrices ${\bf R}_{rk}=(\kappa^{|i-j|})_{1\le i,j\le N}$ and ${\bf R}_{tk}=(\kappa^{|i-j|})_{1\le i,j\le M}$\cite{park2012outage}. For illustration, the pairing policy in \cite{ding2016application,ali2016dynamic} is considered for the distanced-based grouping policy, where the farthest user is paired with the nearest user, and the rest of users are paired in the same fashion. Henceforth, the ranking orders with respect to the distance-based NOMA grouping policy are given by $r_1,r_{\tilde 1},r_2,r_{\tilde 2}=1,4,2,3$. Besides, in the succeeding examples, the technique of the signal alignment is applied to relieve the intra-cell interference.
\begin{table}
  \centering
  \caption{Simulation Parameters.}
\label{tab:list_symb}
\begin{tabular}{|c | c || c | c| }
  \hline
 \textbf{Parameters} & \textbf{Values} & \textbf{Parameters} & \textbf{Values}\\
 \hline
 The number of transmit antennas $M$& 3 &Gaussian noise power $\sigma^2$ & -$99$ dBm\\
 \hline
 The number of receive antennas $N$& 2 & Distances $d_{\tilde k}=2.5d_k$ & 50 m\\
 \hline
 The number of user pairs $K$& 2 & Target rates $R_k=2R_{\tilde k}$ & 1 bps/Hz\\
 \hline
 Path loss exponent $\alpha$& 3.5 & Coverage radius of simulated network& 5 km\\
 \hline
 Intensity of users $\lambda_u$ & $2\times 10^{-4}$ users/m${}^2$ & Exponential correlation coefficient $\kappa $ & 0.9\\
 \hline
 Intensity of BSs $\lambda_b$ & $10^{-5}$ BSs/m${}^2$ & Channel K factor $\mathcal K_{zk}$ & 20 dB\\
 \hline
 Transmit power at BS $P$ & $20$ dBm & 2-norm of the estimated CSI ${\| {\hat {\bf H}_{zk}} \|_F^2}$ & $MN$\\
 \hline
 Interference power $\rho_I$ & $15$ dBm & Power allocation factor $\beta_k^2$ & 0.3\\
\hline
\end{tabular}
\end{table}
\subsection{Verifications}
Fig. \ref{fig:corate} examines the conditional outage probability $p_{k|\mathcal H}$ against the transmission rate $R_{\tilde k}$. In Fig. \ref{fig:corate}, the approximate results of $p_{k|\mathcal H}$ (labeled as ``$k$=1-Approx.'') are obtained by using \eqref{eqn:pkupper}. Clearly, the analytical results are in perfect agreement with the simulation results. In addition, the conditional outage probability increases with the transmission rates. This is due to the tradeoff between the system throughput and reliability. Besides, a slight gap between the approximate and simulated outage probabilities of the near users can be observed. Furthermore, it is not hard to find that the exact outage probability is always upper bounded by the approximate one. This is due to the fact that the strong correlation between ${\rm{SIN}}{{\rm{R}}_{k \to \tilde k}} $ and ${\rm{SIN}}{{\rm{R}}_k}$ renders the near user having a high successful probability to reconstruct its own message if it successfully subtracts the far user's message. The same trends can be observed in Fig. \ref{fig:cokch}.

\begin{figure}
  \centering
  \includegraphics[width=3.5in]{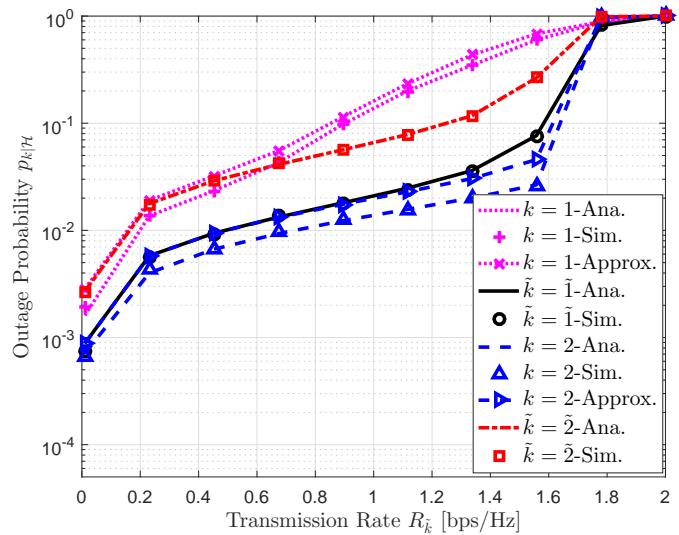}
  \caption{The conditional outage probability $p_{k|\mathcal H}$ versus the transmission rate $R_{\tilde k}$ by setting parameters as $R_{ k}=2R_{\tilde k}$.}\label{fig:corate}
\end{figure}

Fig. \ref{fig:cokch} plots the conditional outage probability $p_{k|\mathcal H}$ against the channel K factor ${\cal K}_{zk}$. As expected, the analytical results coincide with the simulation ones. Moreover, as ${\cal K}_{zk}$ increases, the outage probability decays to an outage floor. It is no doubt that the outage floor is determined by the co-channel interference from surrounding BSs.
\begin{figure}
  \centering
  \includegraphics[width=3.5in]{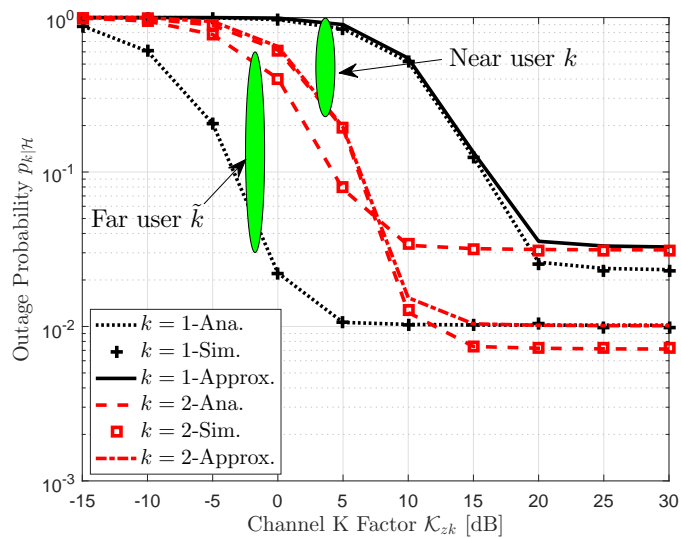}
  \caption{The conditional outage probability $p_{k|\mathcal H}$ versus the channel K factor ${\cal K}_{zk}$ by setting $R_{ k}=2R_{\tilde k}=1$ bps/Hz.}\label{fig:cokch}
\end{figure}

Fig. \ref{fig:aorate} investigates the average outage probability $p_{k}$ against the transmission rate $R_{\tilde k}$. In both Figs. \ref{fig:aorate} and \ref{fig:aolambda}, the exact and approximate outage probabilities of the random grouping policy (labeled as ``Random-Exact'' and ``Random-Approx.'') for far (or near) NOMA user are calculated by using \eqref{eqn:varphi_tdfin} (or \eqref{eqn:varphi_k_rand}) and \eqref{eqn:varphi_ob} (or \eqref{eqn:varphi_obk}), respectively. Moreover, the simulated outage probability (labeled as ``Random-Sim.'') is plotted for verification. Similarly, the numerical results for the distance-based grouping policy can be obtained according to Theorems \ref{the:disb} and \ref{the:dis_k_t}, and the corresponding exact, approximate and simulated outage probabilities are labeled as ``Distance-Exact'', ``Distance-Approx.'' and ``Distance-Sim.'', respectively. From Fig. \ref{fig:aorate}, we observe a perfect match between the analytical and simulation results, which corroborates the validity of our analysis. It can also be seen that the distance-based grouping policy performs better than the random grouping policy in terms of the outage probability of the near user, while the impact on the outage probability of the far user behaves differently. This is due to the fact that the near/far user under the distance-based grouping policy is closer to/farther away from its associated BS than that under random grouping policy. Nevertheless, it is noteworthy that the random grouping policy is frequently used for benchmarking purpose, because it has been intensively investigated in \cite{zhu2018optimal} and \cite{islam2018resource} that the random grouping policy provides the lowest sum rate gain.
\begin{figure}
  \centering
  \includegraphics[width=3.5in]{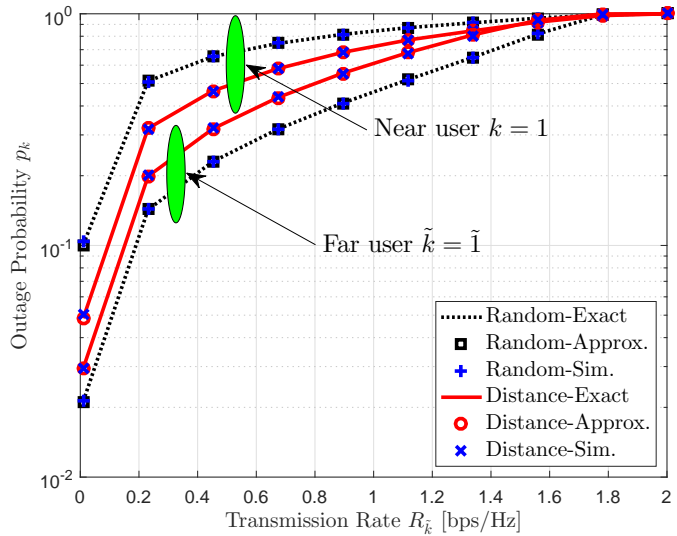}
  \caption{The average outage probability $p_{k}$ versus the transmission rate $R_{\tilde k}$ by setting $R_{ k}=2R_{\tilde k}$.}\label{fig:aorate}
\end{figure}

Fig. \ref{fig:aolambda} depicts the average outage probability $p_{k}$ against the intensity of BSs $\lambda_b$. Aside from the similar findings as Fig. \ref{fig:aorate}, it is evident that the average outage probability is irrelevant to the intensity of BSs under high $\lambda_b$, i.e., interference-limited regime. This is consistent with our analytical results. Nonetheless, the average outage performance declines if $\lambda_b$ falls below a threshold (e.g., $10^{-6}$ BSs/m${}^2$). Indeed, such observations seem to be counterintuitive since the inter-cell interference is alleviated as $\lambda_b$ decreases. However, the decrease of $\lambda_b$ also signifies the expansion of the size of Voronoi cell, on average, this dominant effect engenders severe path loss.
\begin{figure}
  \centering
  \includegraphics[width=3.5in]{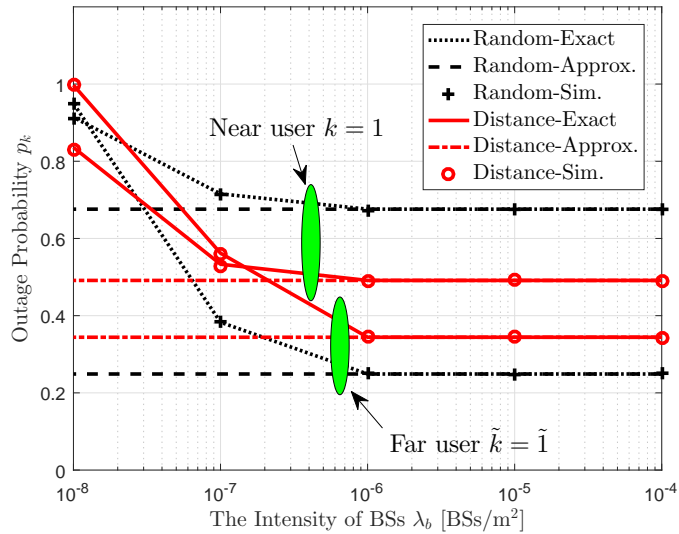}
  \caption{The average outage probability $p_{k}$ versus the transmission rate $R_{\tilde k}$ by setting $R_{ k}=2R_{\tilde k}=1$ bps/Hz.}\label{fig:aolambda}
\end{figure}
\subsection{Impacts of System Parameters}
Fig. \ref{fig:aocorr} illustrates the impact of the exponential correlation coefficient $\kappa$ on the conditional outage probability $p_{k|\mathcal H}$. It is worth mentioning that $\lambda_b=10^{-7}$ BSs/m${}^2$ is designated in this example to weaken the impact of interference so as to better reflect the influence of estimation error correlation. Clearly from Fig. \ref{fig:aolambda}, $\lambda_b=10^{-7}$ BSs/m${}^2$ can be chosen to get rid of interference-limited regime. It is shown that the correlation coefficient has a positive effect on the average outage probability $p_{k|\mathcal H}$ under low channel K factor ${\cal K}_{zk}$. On the contrary, the correlation coefficient impairs the outage performance under high ${\cal K}_{zk}$. Moreover, the conditional outage probability of the far user, i.e., $p_{k|{\cal H}}$, changes very slightly with respect to the correlation coefficient $\kappa$ under ${\cal K}_{zk}=20$dB. This is due to the fact that the channel estimation error diminishes as $\mathcal K_{zk}$ increases. As a consequence, the error correlation under high $\mathcal K_{zk}$ has a negligible impact on the outage probability compared to that of the interference, especially for the far user with larger path loss.

\begin{figure}
  \centering
  \includegraphics[width=3.5in]{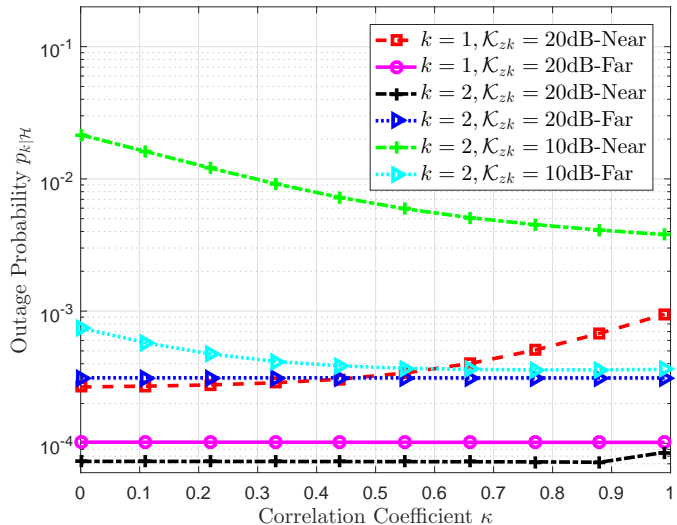}
  \caption{The conditional outage probability $p_{k|\mathcal H}$ versus the exponential correlation coefficient $\kappa$ with parameters $R_{ k}=2R_{\tilde k}=1$bps/Hz and $\lambda_b=10^{-7}$BSs/m${}^2$.}\label{fig:aocorr}
\end{figure}

As shown in Fig. \ref{fig:upp}, the conditional outage probability $p_{k|\mathcal H}$ of NOMA pair $k=1$ is plotted against the transmission rate $R_{\tilde k}$ under different channel K factor ${\mathcal K}_{zk}$. With the increase of ${\mathcal K}_{zk}$, the outage curves become steeper. In particular, the outage probability significantly increases to 1 for ${\mathcal K}_{zk}=50$ dB as $R_{\tilde k}$ grows from 1.73 bps/Hz to 1.74 bps/Hz. However, this is not beyond our expectation which can be elucidated as follows. As unveiled by the asymptotic outage analysis in Section \ref{sec:asy}, the upper bound of the transmission rate $R_{\tilde k}=1.737$ bps/Hz, which is in accordance with our observations in Fig. \ref{fig:upp}. In contrast, it is found from Theorem \ref{the:upp2} that $R_1$ and $R_2$ are upper bounded by 21.846 bps/Hz and 23.758 bps/Hz, respectively. Apparently, the outage performance of NOMA users are considerably restricted by the target rates of the far users, i.e., $R_{\tilde k}$.
\begin{figure}
  \centering
  \includegraphics[width=3.5in]{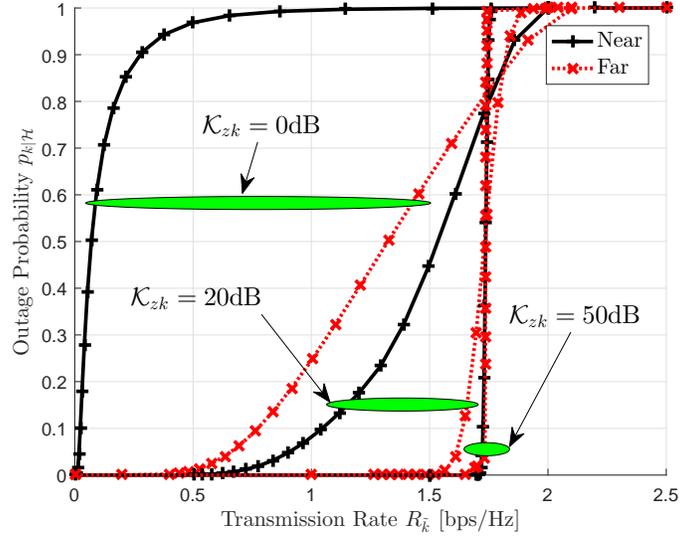}
  \caption{The conditional outage probability $p_{k|\mathcal H}$ versus the transmission rate $R_{\tilde k}$ by setting $R_{ k}=2R_{\tilde k}$ and $\lambda_b=0$~BSs/m${}^2$.}\label{fig:upp}
\end{figure}
\subsection{Optimal System Design}
Fig. \ref{fig:opt} compares the maximum achievable goodput of the proposed MIMO-NOMA enhanced scheme with those of three benchmarking schemes, including MIMO-OMA with precoding scheme \cite{park2012outage}, MIMO-OMA without precoding scheme and MIMO-NOMA without precoding scheme \cite{ding2016application,cui2018outage}. For comparative fairness, the coefficient $\beta_k^2$ is also used for the allocation of the orthogonal resources (time/bandwidth) in the OMA scheme. Additionally, it is worthwhile to note that the approximate expression of ${p_{k|{\cal H}}}$ is used to solve \eqref{eqn:prob_defsimpdec} to realize efficient and robust system design. It can be seen from Fig. \ref{fig:opt} that the proposed scheme performs the best particularly for a great difference between $d_k$ and $d_{\tilde k}$. This is due to the fact that the intrinsic nature of NOMA exploits the difference between channel gains to boost the spectral efficiency. Moreover, the precoding schemes can achieve a noticeable goodput gain over those without precoding. In addition, as the channel K factor ${\cal K}_{zk}$ increases (till no estimation error), the maximum achievable goodput increases up to an upper bound, which obviously depends on the severity of co-channel interference, i.e., $\lambda_b$.
\begin{figure}
  \centering
  \includegraphics[width=3.5in]{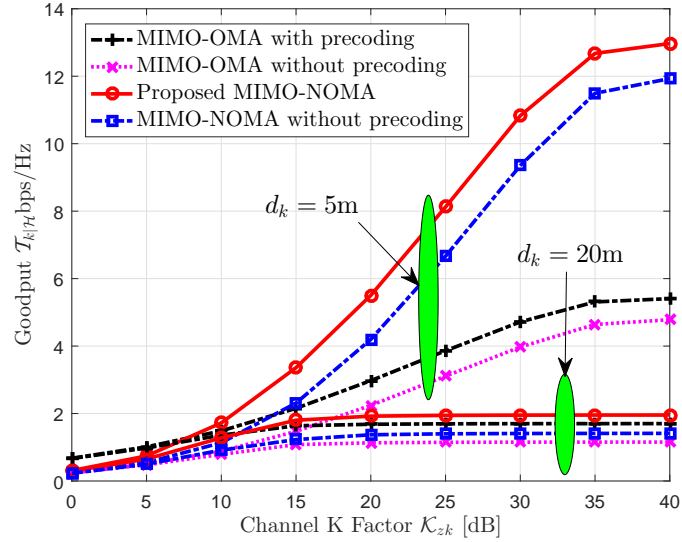}
  \caption{Performance comparison with the three baseline schemes in terms of the goodput ${\cal T}_{k|\mathcal H}$ versus the channel K factor ${\cal K}_{zk}$ with $\varepsilon=10^{-2}$.}\label{fig:opt}
\end{figure}

\section{Conclusions}\label{sec:con}
The performance of MIMO-NOMA enhanced SCNs has been thoroughly investigated by taking into account the assumption of imperfect CSI. The channel estimation error and the spatial randomness of BSs have been characterized by capitalizing on Kronecker correlation model and homogeneous PPP, respectively. The outage probabilities of MIMO-NOMA systems have been derived in compact form by considering two distinct NOMA grouping policies, including the random grouping and the distance-based grouping. It has been proved that the average outage probabilities are independent of the intensity of BSs in the interference-limited regime. Interestingly, it has been found that the outage performance is degraded under a sufficiently low intensity of BSs. Besides, the asymptotic analyses have been carried out to gain more insights into the outage behavior under high quality of CSI. Specifically, it has been shown that the target rates must be limited up to a bound to achieve an arbitrarily low outage probability in the absence of inter-cell interference. The outage expressions have empowered the maximization of the goodput by properly designing the precoding matrix, receiver filters and transmission rates. At last, the numerical results have also revealed that the estimation error correlation contributes to improve the outage performance under low quality of CSI. Whereas, opposite observations have been made under high quality of CSI. Furthermore, for intelligent resource allocation of NOMA enhanced SCNs, we plan to extend the machine-learning based approaches as in \cite{zhou2019reliable} and \cite{liao2019learning} to our future work.

\bibliographystyle{ieeetran}
\bibliography{manuscript_1}

\end{document}